\documentclass[reqno,11pt]{amsart}
\usepackage{geometry}
\usepackage{mathtools,amssymb,amsthm,mathrsfs,color,lineno,paralist,graphicx,float}
\usepackage{mathabx}
\usepackage{constants}
\newcommand{\parenthezises}[1]{\arabic{#1}}
\newconstantfamily{sig}{
	symbol=\sigma,
	format=\parenthezises,
}
\newconstantfamily{K}{
	symbol=K,
	format=\parenthezises,
}
\numberwithin{equation}{section}
\usepackage[colorlinks,
linkcolor=red,
anchorcolor=green,
citecolor=blue, 
]{hyperref}

\usepackage{enumitem}
\usepackage[T1]{fontenc}
\usepackage[utf8]{inputenc}

\usepackage{calc}
\definecolor{bleu1}{RGB}{0,57,128}
\def\bleu1{\color{bleu1}}

\usepackage{etoolbox}
\patchcmd{\section}{\normalfont}{\normalfont \bleu1}{}{}
\patchcmd{\subsection}{\normalfont}{\normalfont \bleu1}{}{}
\patchcmd{\subsubsection}{\normalfont}{\normalfont \bleu1}{}{}

\usepackage{xpatch}
\xpatchcmd{\proof}{\itshape}{\it \bleu1 \proofnamefont}{}{}
\newcommand{\proofnamefont}{}

\newtheorem{proposition}{Proposition}[section]
\newtheorem{theorem}{Theorem}[section]
\newtheorem{definition}{Definition}[section] 
\newtheorem{lemma}{Lemma}[section]
\newtheorem{remark}{Remark}[section]
\newtheorem{corollary}{Corollary}[section]

\newcommand{\Z}{{\mathbb Z}}
\newcommand{\R}{{\mathbb R}}

\usepackage{tikz,tikz-3dplot,tikz-cd}
\usepackage{pgfplots}
\usepgfplotslibrary{polar}

\begin{document}
	
\title[]{Winding number, density of states and acceleration}

\author{Xueyin Wang}
\address{
Chern Institute of Mathematics and LPMC, Nankai University, Tianjin 300071, China
}

\email{xueyinwang@mail.nankai.edu.cn}

\author{Zhenfu Wang}
\address{
	Chern Institute of Mathematics and LPMC, Nankai University, Tianjin 300071, China
}

\email{zhenfuwang@mail.nankai.edu.cn}

\author {Jiangong You}
\address{
Chern Institute of Mathematics and LPMC, Nankai University, Tianjin 300071, China} \email{jyou@nankai.edu.cn}

\author{Qi Zhou}
\address{
Chern Institute of Mathematics and LPMC, Nankai University, Tianjin 300071, China
}

\email{qizhou@nankai.edu.cn}

\begin{abstract}
Winding number and density of states are two fundamental physical quantities for non-self-adjoint quasi-periodic Schr\"odinger operators, which reflect the asymptotic distribution of zeros of the characteristic determinants of the truncated operators under Dirichlet boundary condition,  with respect to complexified phase and the energy respectively.  We will prove that the winding number is in fact Avila's acceleration and it is also closely related to the density of states by a generalized Thouless formula for non-self-adjoint Schr\"odinger operators and Avila's global theory. 
\end{abstract}

\maketitle

\section{Introduction}

We study the analytic quasi-periodic Schr\"odinger operators 
\begin{equation}\label{sch}
	(H\psi)_{n} = \psi_{n+1}+\psi_{n-1}+v(x+n\alpha)\psi_{n}=E\psi_{n},
\end{equation}
where $x\in\mathbb{T}^{d}:=\mathbb{R}^{d}/(2\pi\mathbb{Z})^{d}$ is the phase, $v\in C^{\omega}(\mathbb{T}^{d},*)$ is the potential, $\alpha$ is the frequency with $(1,\alpha)$ rationally independent. Here $*$ denotes $\mathbb{R}$ or $\mathbb{C}$, then \eqref{sch} defines a self-adjoint or non-self-adjoint  quasi-periodic Schr\"odinger operator (non-Hermitian quasi-crystals in physical literature) on $\ell^2(\mathbb{Z})$. 

In the self-adjoint case, ever since the almost-periodic flu around the 1980s \cite{Sim2},  due to its importance in physics \cite{AOS, TKNN}  and surprising spectral richness \cite{Av0, Av1, AK06, AYZ1, BG, GS11, J, JL08}, many breakthroughs have appeared in recent years (readers are invited to consult \cite{J1} and the reference therein for details). On the other hand, while Hermiticity lies at the heart of quantum mechanics, recent experimental advances in controlling dissipation have brought about unprecedented flexibility in engineering non-Hermitian Hamiltonians in open classical and quantum systems \cite{Gong}; Non-Hermitian Hamiltonians exhibits rich phenomena without Hermitian counterparts \cite{AGU,BBK}, while all of these phenomena can be exhibited in  non-Hermitian quasi-crystals \cite{longhi,jiang2019interplay}. In this paper, we are more interested in the non-self-adjoint case.

A fruitful way to study \eqref{sch} is by method of quasi-periodic $\mathrm{SL}(2,\mathbb{C})$ cocycle. Recall that the eigenvalue equations $H\psi=E\psi$ are equivalent to a certain family of the discrete dynamical systems, so called Schr\"odinger cocycle $(\alpha,A)\in \mathbb{T}^{d}\times  C^{\omega}( \mathbb{T}^{d}, {\rm SL}(2,\mathbb{C}))$, i.e., 
\begin{equation*}
	\begin{pmatrix}
		\psi_{n+1}\\ \psi_{n}
	\end{pmatrix}=A(x+n\alpha )\begin{pmatrix}
		\psi_{n}\\ \psi_{n-1}
	\end{pmatrix}, \ \text{where} \ A(x)=S_{E,v}(x)=\begin{pmatrix}
		E- v(x)&-1\\1&0
	\end{pmatrix}.
\end{equation*}
The celebrated Avila's global theory gives a classification of one-frequency analytic quasi-periodic $\mathrm{SL}(2,\mathbb{C})$ cocycles by two dynamical invariants: Lyapunov exponent and acceleration. Let us introduce these two invariants: 
Lyapunov exponent is defined by
\begin{equation*}
	L(E,\mathrm{i}y):=L(\alpha,A(\cdot+\mathrm{i}y)):=\lim_{n\rightarrow \infty}\frac{1}{(2\pi)^{d}}\int_{\mathbb{T}^{d}}\frac{1}{n}\log \|M_{n}(x+\mathrm{i}y)\| \mathrm{d}x,
\end{equation*}
where $M_{n}$ is the $n$-step transfer matrix given by 
\begin{equation*}
	M_{n}(x+\mathrm{i}y)=A(x+\mathrm{i}y+n\alpha)\cdots A(x+\mathrm{i}y+\alpha).
\end{equation*}
  In the one-frequency case $(d=1)$,   according to Avila's global theory \cite{Av0}, Lyapunov exponent is a piecewise linear and convex function in $y$, and one can define the acceleration of $(\alpha, A_{\mathrm{i}y}):=(\alpha, A(\cdot+\mathrm{i}y))$ as:
\begin{equation*}
	\omega^{\pm}(E,\mathrm{i}y):=\omega^{\pm}(\alpha, A_{\mathrm{i}y}):=\lim _{\epsilon \rightarrow 0^{\pm}} \frac{1}{\epsilon} (L(\alpha, A_{\mathrm{i}(y+\epsilon)})-L(\alpha, A_{\mathrm{i}y})),
\end{equation*}
i.e. the right (left) derivatives of $L(E,\mathrm{i}y)$ with respect to $y$. The key observation of Avila is that the acceleration  is quantized: $\omega (\alpha, A_{\mathrm{i}y}) \in \mathbb{Z}$. 
This observation  plays a fundamental role in his global theory  \cite{Av0}.

On the other hand, in studying \eqref{sch}, no matter in the physical setting  \cite{AGU,BBK,Gong,jiang2019interplay,LZC,longhi} or in the mathematical setting \cite{Av1, AJ,BG,GS08,GS11,GSV19,J,JL08}, a key object is the  determinant of truncation with Dirichlet boundary condition
\begin{eqnarray*}
	f_{n}(x,E)&:=&\det(H_{n}(x)-E)\\
	&:=& \det \begin{pmatrix}
		v(x+\alpha)-E&1&0&\cdots&0\\
		1&v(x+2\alpha)-E&1&\cdots&0\\
		\vdots&\ddots&\ddots&\ddots&\vdots\\
		0&\vdots&\ddots&\ddots&1\\
		0&\cdots&0&1&v(x+n\alpha)-E
	\end{pmatrix}.
\end{eqnarray*}

If $v$ is complex, then in studying the distribution of the zeros of $f_{n}(x,E)$, two fundamental physical quantities appears, i.e. winding number and density of states.  In this paper, we will prove that the winding number and the density of states are closely related and Avila's acceleration establishes the bridge between them.

\subsection{Winding number and topological phase transition}
Let us introduce the first physical quantity called \textit{winding number} \cite{Gong,LZC,longhi}:
\begin{equation*}
	\nu^{\pm}(E,\mathrm{i}y)=\lim_{\epsilon\rightarrow 0^{\pm}}\lim_{n\rightarrow\infty}\nu_{n}(E,\mathrm{i}(y+\epsilon)),
\end{equation*} 
where
\begin{equation*}
	\nu_{n}(E,\mathrm{i}y)=\frac{-1}{2\pi\mathrm{i}n}\int_{0}^{2\pi} \frac{\partial}{\partial x} \log \det(H_{n}(x+\mathrm{i}y)-E) \mathrm{d}x.
\end{equation*}
One can prove that the winding number is an integer, which represents the topological phase in one-dimensional non-Hermitian lattices. Therefore,  topological phase transition of the non-Hermitian systems can be observed via detecting the variation of winding number \cite{Kita10}. 

Let us first explain the physical meaning of winding number. Let $E_{j}(x+\mathrm{i}y), j=1,\cdots,n$ be the eigenvalues of $H_{n}(x+\mathrm{i}y)$. Then $\det (H_{n}(x+\mathrm{i}y)-E)=\prod_{j=1}^{n}(E_{j}(x+\mathrm{i}y)-E)$ implies that  the winding number can be expressed as
\begin{equation*}
	\nu_{n}(E,\mathrm{i}y)=\frac{1}{n}\sum_{j=1}^{n}\frac{-1}{2\pi}\int_{0}^{2\pi} \frac{\partial}{\partial x} \arg(E_{j}(x+\mathrm{i}y)-E)\mathrm{d}x,
\end{equation*}
where $\arg (E_{j}(x+\mathrm{i}y)-E)$ is the argument of the complex energy $E_{j}(x+\mathrm{i}y)-E$.  
In this sense, a non-trivial winding number gives the number of times the complex eigenenergies encircle the base point $E$. 

To make clearly the relation of winding number and acceleration, we consider the complex extension of \eqref{sch}:
\begin{equation}\label{compentend}
	(H_{v,y}\psi)_{n}=\psi_{n+1}+\psi_{n-1}+v(x+\mathrm{i}y+n\alpha)\psi_{n},
\end{equation}
and denote by $\Sigma_{v,y}$ the spectrum of $H_{v,y}$. As explained in \cite{Gong},  the winding number $\nu^-(E,\mathrm{i}y)=\nu^+(E,\mathrm{i}y)$ if the base point $E\notin\Sigma_{v,y}$, i.e. $(\alpha, S_{E,v})$ is uniformly hyperbolic\footnote{One may consult Proposition \ref{equi} for this.}.  For $E\in \Sigma_{v,y}$ with positive Lyapunov exponent, we have 

\begin{theorem}\label{windacc}
	For any $(\alpha,v)\in \mathbb{R}\backslash\mathbb{Q}\times C^{\omega}(\mathbb{T}, \mathbb{C})$, if $L(E,\mathrm{i}y)>0$, then
	\begin{equation}\label{wiac}
		\begin{split}
			\nu^{+}(E,\mathrm{i}y)&=\omega^{+}(E,\mathrm{i}y),\\
			\nu^{-}(E,\mathrm{i}y)&=\omega^{-}(E,\mathrm{i}y).
		\end{split}
	\end{equation}
	In particular, if $E\notin\Sigma_{v,y}$, then
	\begin{equation*}
		\nu^{+}(E,\mathrm{i}y)=\nu^{-}(E,\mathrm{i}y)=\omega^{+}(E,\mathrm{i}y)=\omega^{-}(E,\mathrm{i}y).
	\end{equation*}
\end{theorem}

\begin{remark}
	\eqref{wiac} was first observed in \cite{LZC}, with a physics level of proof, here we provide a rigorous proof.
\end{remark}

As we said, the significance of winding number indicates the topological phase transition in non-Hermitian Hamiltonian. For some models, the transition points can  be exactly located. For example, consider the following simple but well-studied physical model  \cite{jiang2019interplay,LZC,longhi}:
	\begin{equation}\label{namo}
			(H_{2\lambda\cos,y}\psi)_{n}=\psi_{n+1}+\psi_{n-1}+2\lambda\cos(x+\mathrm{i}y+n\alpha) \psi_{n},
		\end{equation}
		i.e. almost Mathieu operator with complex extension. Then we have the following:
\begin{corollary}\label{amo}
	Let $\alpha\in\mathbb{R}\backslash\mathbb{Q}$, $\lambda\in(0,1)$.
	Then  $-\log\lambda$ is the topological transition point of \eqref{namo}. More precisely, we have the following:
	
	\begin{enumerate}[font=\normalfont, label={(\arabic*)}]
			\item  If $0<y<-\log\lambda$, then $\nu^{+}(E,\mathrm{i}y)=\nu^{-}(E,\mathrm{i}y)=0$ for any $E\notin \Sigma_{2\lambda \cos,0}$.
			\item If $y>-\log\lambda$, then there exists $E\notin \Sigma_{2\lambda \cos,0}$ with  $\nu^{+}(E,\mathrm{i}y)=\nu^{-}(E,\mathrm{i}y)=1$.
		\end{enumerate}	
\end{corollary}

\begin{remark}
	It was previously proved in \cite{WYZ} that if $0<y<-\log\lambda$, then $\Sigma_{2\lambda\cos,0}=\Sigma_{2\lambda\cos,y}$. 
\end{remark}

To analyze the topological transition for the general one-frequency analytic quasi-periodic Schr\"odinger operators, one can apply the quantitative global theory developed by Ge-Jitomirskaya-You-Zhou \cite{GJYZ} to study the transition points.

Our second application concerns the distribution of zeros of $\det (H_{n}(z)-E)$ with respect to the phase $z$. As was first claimed by  Goldstein-Schlag  \cite{GS08}: ``In fact, these zeros turn out to be relatively uniformly distributed along certain circles.''  Here, we exactly give the location of these circles. Denote by $\mathcal{A}(r_{1},r_{2})$ the annulus $\{z\in\mathbb{C}:r_{1}<|z|<r_{2}\}$, then $f_{n}(E,z)$ is holomorphic in $z$ on the annulus $\mathcal{A}(\mathrm{e}^{-h},\mathrm{e}^{h})$ if $v$ is analytic on the strip $\mathbb{T}_{h}$ where $\mathbb{T}_{h} = \{x+{\rm i}y:  x \in \mathbb{T}, y\in (-h,h)\}$. Let $\mathcal{Z}_{f_{n}}=\{z\in\mathbb{C}: f_{n}(z)=0\}$, and denote the number of zeros (counted with multiplicities) by
\begin{equation*}
	N_{n}(r_{1},r_{2})=\#\{z\in\mathcal{Z}_{f_{n}}: z\in\mathcal{A}(r_{1},r_{2})\}.
\end{equation*}
Then we show the winding number actually counts the averaged number of zeros in the annulus:
\begin{corollary}\label{zeros}
	Let $(\alpha,v)\in \mathbb{R}\backslash\mathbb{Q}\times C^{\omega}(\mathbb{T}_{h},\mathbb{C})$. Suppose $L(E,\mathrm{i}y)>0$ for any $y\in[y_{1},y_{2}]$. Then we have 
	\begin{enumerate}[font=\normalfont, label={(\arabic*)}]
		\item \label{item:zeros1} 
		\begin{equation}\label{z}
			\lim_{\epsilon\rightarrow0^{+}}\lim_{n\rightarrow\infty} \frac{1}{n}N_{n}(\mathrm{e}^{-(y_{2}+\epsilon)},\mathrm{e}^{-(y_{1}-\epsilon)})=\nu^{+}(E,\mathrm{i}y_{2})-\nu^{-}(E,\mathrm{i}y_{1}).
		\end{equation}
		Especially, 
		\begin{equation*}
			\lim_{\epsilon\rightarrow0^{+}}\lim_{n\rightarrow\infty} \frac{1}{n}N_{n}(\mathrm{e}^{-(y+\epsilon)},\mathrm{e}^{-(y-\epsilon)})=\nu^{+}(E,\mathrm{i}y)-\nu^{-}(E,\mathrm{i}y).
		\end{equation*}
		\item \label{item:zeros2} Furthermore,
		\begin{equation*}
			\lim_{\epsilon\rightarrow0^{+}}\lim_{n\rightarrow\infty} \frac{1}{n} N_{n}(\mathrm{e}^{-\epsilon},\mathrm{e}^{\epsilon})=2\nu^{+}(E,\mathrm{i}0),
		\end{equation*}
	if $v\in C^{\omega}(\mathbb{T},\mathbb{R})$,  $E\in \R$ with $L(E,\mathrm{i}0)>0$.
	\end{enumerate}
\end{corollary}

Indeed, simple calculation shows that 
\begin{equation}\label{count}
	\begin{split}
		\nu_{n}(E,\mathrm{i}y_{2})-\nu_{n}(E,\mathrm{i}y_{1})&=\frac{1}{2\pi\mathrm{i}n}\oint_{|z|=\mathrm{e}^{-y_{1}}}\frac{f_{n}'(z)}{f_{n}(z)} \mathrm{d}z-\frac{1}{2\pi\mathrm{i}n}\oint_{|z|=\mathrm{e}^{-y_{2}}}\frac{f_{n}'(z)}{f_{n}(z)} \mathrm{d}z\\
		& =\frac{1}{n}N_{n}(\mathrm{e}^{-y_{2}},\mathrm{e}^{-y_{1}}) + \frac{1}{2n}\#\{z\in\mathcal{Z}_{f_{n}}: z\in\partial \mathcal{A}(\mathrm{e}^{-y_{2}},\mathrm{e}^{-y_{1}})\}.
	\end{split}
\end{equation}
Then Corollary \ref{zeros} follows immediately from \eqref{count} and Theorem \ref{windacc}.

Note, by Avila's global theory \cite{Av0}, $L(E,\mathrm{i}y)$ is convex, piecewise linear in $y$, and we denote by $\{\gamma_{j}\}$ the set of turning points.  Corollary \ref{zeros} implies that these so-called ``certain circles'' are concentric circles and the radii of the circles are just $\{\mathrm{e}^{-\gamma_{j}}\}$, see the diagram in Figure \ref{circletikz}. Moreover, if $v\in C^{\omega}(\mathbb{T},\mathbb{R})$, then by the symmetry of zeros the number of the concentric circles is odd when $E\in\Sigma_{v,y}$ and it is even when $E\in\mathbb{R}\backslash\Sigma_{v,y}$. 

\begin{figure}[htbp]
	\centering
	\begin{tikzpicture} 
		\draw (0,0) circle (2cm);
		\draw[mark=ball,mark repeat=1.5,smooth,domain = 0: 2*pi] plot ({1.5*cos(\x r)},{1.5*sin(\x r)});
		\draw[mark=ball,mark repeat=2,smooth,domain = 0: 2*pi] plot ({1*cos(\x r)},{1*sin(\x r)});
		
		\draw[densely dotted, black] (0,0) node[below] {}-- (0,2) node[below left] {$1$};
		\draw[densely dotted, black] (0,0) node[below] {}-- (1.0608,1.0608) node[left] {$\mathrm{e}^{-\gamma_{1}}$};
		\draw[densely dotted, black] (0,0) node[below] {}-- (1,0) node[below left] {$\mathrm{e}^{-\gamma_{2}}$};
		\fill[red] (0,0) ellipse (1pt and 1pt);
	\end{tikzpicture}
	\qquad \qquad
	\begin{tikzpicture}
		\draw[->] (0,0) -- (3,0) node[below] {$y$};
		\draw[->] (0,0) -- (0,3.5) node[right] {$L(E,\mathrm{i}y)$};
		\draw [black, thick] (0,1) node {}-- (1,1) node [above] {};
		\draw [black, thick] (1,1) node {}-- (1.75,1.75) node [above] {};
		\draw [black, thick] (1.75,1.75) node {}-- (2.25,2.75) node [above] {};
		\fill[black] (1,1) circle (0.4pt);
		\fill[black] (1.75,1.75) circle (0.4pt);
		\draw[densely dotted, black] (1,0) node[below] {$\gamma_{1}$}-- (1,1) node {};
		\draw[densely dotted, black] (1.75,0) node[below] {$\gamma_{2}$}-- (1.75,1.75) node {};
		\draw[] (0,0) node[below] {$0$}-- (0,0) node {};
	\end{tikzpicture}
	\caption{Location of the circles.}
	\label{circletikz}
\end{figure}
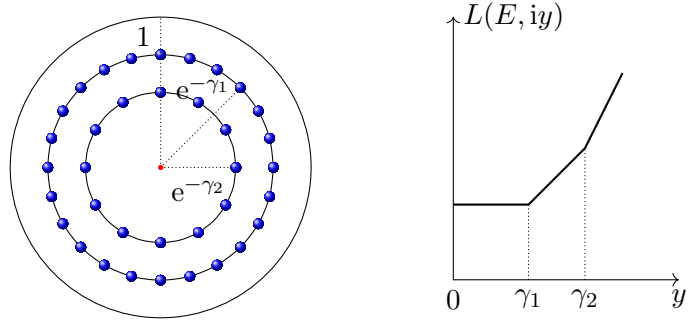

\begin{remark}
	A prior to us, by  developing potential theory on annulus, 
	Han-Schlag \cite{HS22}  gave a quantitative estimates of  \eqref{z} assuming $\alpha \in \mathrm{SDC}\footnote{If there exists 
		$\kappa>0$ and $\tau>1$, such that 
		\begin{equation*}
			\|\mathbf{k}\alpha\|_{\mathbb{T}}\geq \frac{\kappa}{|\mathbf{k}|(\log |\mathbf{k}|)^{\tau}},  \  \text{for all}  \  \mathbf{k}\neq 0.
	\end{equation*}}:$ 
	\begin{equation*}
		\bigg|\frac{1}{n}N_{n}(\mathrm{e}^{-(y_{2}+\epsilon)},\mathrm{e}^{-(y_{1}-\epsilon)})- \nu^{+}(E,\mathrm{i}y_{2})+ \nu^{-}(E,\mathrm{i}y_{1})  \bigg|\leq  C\epsilon^{-2}n^{-c}.
	\end{equation*}
\end{remark}

\subsection{Density of states and Thouless formula} 
Let us consider the second physical quantity: the density of states, which describes the limit distribution of the energy spectrum. For any $x\in\mathbb{T}^{d}$, the determinant $f_{n}(E,x+\mathrm{i}y)$ is a polynomial in $E$ of degree $n$. Its roots $E_{j}(x+\mathrm{i}y), j=1,\cdots,n$ are exactly the eigenvalues of the  truncated operator $H_{n}(x+\mathrm{i}y)$. Denote the normalized eigenvalues counting measure of $H_{n}(x+\mathrm{i}y)$ by
\begin{equation*}
	\mathrm{d}\mathcal{N}_{n}^{x+\mathrm{i}y}=\frac{1}{n} \sum_{j=1}^{n} \delta_{E_{j}(x+\mathrm{i}y)}.
\end{equation*}
Then the density of states measure is defined as the weak limit of $\mathrm{d}\mathcal{N}_{n}^{x+\mathrm{i}y}$ as $n\rightarrow \infty$, which exists and is same for a full Lebesgue measure set of  $x\in\mathbb{T}^{d}$, we denote it by $\mathrm{d}\mathcal{N}^{\mathrm{i}y}$. In the self-adjoint case ($v$ is real-valued and $y=0$), we denote it as $\mathrm{d}\mathcal{N}$.

In the early 1970s, for the one-dimensional Anderson model, Thouless \cite{Thou72}  observed that (with a physical level proof) 
\begin{equation*}
	L(E)= \int_{\mathbb{C}} \log |\tilde{E}-E|\mathrm{d}\mathcal{N}(\tilde{E}),
\end{equation*}
which is known as Thouless formula in the literature. This formula 
relates a mathematical physics object, the density of states, and a dynamical systems object, the Lyapunov exponent of the Schr\"odinger cocycle, via Hilbert transform.  
In this paper, we extend the Thouless formula for the non-self-adjoint quasi-periodic Schr\"odinger operators. To state the result, we denote by $\mathrm{d}m$ the  Lebesgue measure in complex plane $\mathbb{C}$  and 
\begin{equation*}
	\mathcal{Z}_{L,y}=\{E\in\mathbb{C}: L(E,\mathrm{i}y)=0\}.
\end{equation*}
Let  $\mathbb{T}^{d}_{h} = \mathbb{T}_{h}\times \cdots \times \mathbb{T}_{h}$ be the product torus with analytic extension. Then we have the following:
\begin{theorem}\label{open}
	Let $v\in C^{\omega}(\mathbb{T}_{h}^{d},\mathbb{C})$,  $\alpha\in\mathbb{T}^{d}$ with $(1,\alpha)$ rationally independent. Suppose that  $m(\mathcal{Z}_{L,y})=0$. Then 
	\begin{enumerate}[font=\normalfont, label={(\arabic*)}]
		\item \label{open1} (Weak convergence) For a.e. $x\in \mathbb{T}^{d}$, 
		\begin{equation*}
			\mathrm{d}\mathcal{N}^{x+\mathrm{i}y}_{n}\rightharpoonup \mathrm{d}\mathcal{N}^{\mathrm{i}y} :=\frac{1}{2\pi} \Delta L(E,\mathrm{i}y) \mathrm{d}m ,\quad \text{as} \ n\rightarrow \infty,
		\end{equation*}
		where $\Delta$ is the distributional Laplacian with respect to $E$.
		\item\label{open2} (Thouless formula) For every $E\in \mathbb{C}$,
		\begin{equation}\label{thou}
			L(E,\mathrm{i}y)= \int_{\mathbb{C}} \log |\tilde{E}-E|\mathrm{d}\mathcal{N}^{\mathrm{i}y}(\tilde{E}).
		\end{equation}
	\end{enumerate}	
\end{theorem}

\begin{remark}
	The two-dimensional spectrum, i.e. $m(\mathcal{Z}_{L,y})>0$,  does exist for some  non-self-adjoint quasi-periodic Jacobi operators and  the Thouless formula \eqref{thou} does not hold in this case. See our forthcoming paper \cite{WWYZ2}.
\end{remark}

Let us first explain motivations for  establishing the Thouless formula for non-self-adjoint operators. 

First, in non-Hermitian physics, Thouless formula is important, since it provides an approximated expression of localization length (inverse of Lyapunov exponent) at complex spectrum \cite{longhi, Cai21}, i.e., for sufficiently large $n$
\begin{equation*}
	L(E,\mathrm{i}y)\approx \frac{1}{n}\sum_{j=1}^{n} \log |E_{j}(x+\mathrm{i}y)-E|,
\end{equation*}
where $E_{j}(x+\mathrm{i}y)$ are eigenvalues of $H_{n}(x+\mathrm{i}y)$.

The  second motivation comes from the recent study of  \textit{quantitative} global theory of one-frequency Schr\"odinger operators \cite{GJYZ}, i.e. studying \eqref{compentend} by studying its dual operator 
\begin{equation*}
	(\mathcal{L}u)(n)=\sum\limits_{k=-d}^{d} \mathrm{e}^{-k y}\hat{v}_k u_{n+k}+2\cos(\theta+n\alpha)u_n, \ \ n\in\Z.
\end{equation*}
The core analysis is reduced to establish the Thouless-type formula for \eqref{compentend} and dual operator simultaneously. In the self-adjoint case (i.e. $y=0$), it was studied by Haro-Puig \cite{HP13}, however their method can't be generalized to the non-self-adjoint setting due to lack of spectral theorem. 

Now let us back to the progress of Thouless formula. In the one-dimensional case, the first rigorous proof of the Thouless formula was given by Avron-Simon \cite{AS83}. Later Craig-Simon \cite{CS83b} provided a simple proof by using subharmonicity of the Lyapunov exponent.  In the multi-dimensional case, Craig-Simon \cite{CS83a} proved the Thouless formula for the Schr\"odinger operators with bounded ergodic stationary potential on a strip. In \cite{KS88}, Kotani-Simon proved the Thouless formula for the general bounded measurable ergodic Schr\"odinger operators on a strip. Finally, a more general version of the Thouless formula for the self-adjoint block Jacobi matrices with dynamically defined entries was obtained by Chapman-Stolz \cite{CS15}.  Thouless formula automatically implies the density of states measure is log-H\"older \cite{CS83b} continuous, which turns out to be optimal for self-adjoint quasi-periodic Schr\"odinger operator \cite{ALSZ}. 

All the above results focus on the self-adjoint Schr\"odinger operators, i.e., the potential is real. If the potential of the Schr\"odinger operator is complex, as far as we know, the only existed result was due to Goldsheid-Khoruzhenko \cite{GK05}, where they assumed the potential is a complex i.i.d. random sequence. 

\subsection{Relationship, methods of the proof}

Theorem \ref{windacc} and Theorem \ref{open} describe the distribution of zeros of $\det (H_{n}(z)-E)$ with respect to the phase $z$ and the energy $E$.  These two kinds of distribution are  nicely related in the following formula:

\begin{corollary}\label{relation}
	Let $(\alpha,v)\in\mathbb{R}\backslash\mathbb{Q} \times C^{\omega}(\mathbb{T}_{h},\mathbb{R})$. Suppose $|y|<h$ and $L(E,\mathrm{i}0)> 0$. Then
	\begin{equation}\label{accden}
		\int_{0}^{y} \nu^{+}(E,\mathrm{i}\tilde{y}) \mathrm{d}\tilde{y}+L(E,\mathrm{i}0)= \int_{\mathbb{C}} \log |\tilde{E}-E|\mathrm{d}\mathcal{N}^{\mathrm{i}y}(\tilde{E}).
	\end{equation}
\end{corollary}

One can view $L(E,\mathrm{i}y)$ as a function with two variables $E$ and $y$. The L.H.S. of \eqref{accden} is in fact the integral representation of $L(E,\mathrm{i}y)$ as  a   piecewise linear convex function in  $y$, while  the R.H.S. of \eqref{accden} is in fact  the Riesz representation of the $L(E,\mathrm{i}y)$ as  a subharmonic function in $E$.

The proof of Theorem \ref{windacc} and Theorem \ref{open} are reduced to the following theorem.
\begin{theorem}\label{converge}
	Let $v\in C^{\omega}(\mathbb{T}^{d}_{h},\mathbb{C})$, $\alpha\in\mathbb{T}^{d}$ with $(1,\alpha)$ rationally independent. Suppose that $L(E,\mathrm{i}y)>0$. Then for a.e. $x\in\mathbb{T}^{d}$,
	\begin{equation}\label{fiLE}
		\lim_{n\rightarrow\infty} \frac{1}{n}\log|f_{n}(E,x+\mathrm{i}y)|=\lim_{n\rightarrow\infty}\frac{1}{(2\pi)^{d}}\int_{\mathbb{T}^{d}}\frac{1}{n} \log|f_{n}(E,x+\mathrm{i}y)|\mathrm{d}x=L(E,\mathrm{i}y).
	\end{equation}
\end{theorem}

Before explaining more on Theorem \ref{converge}, let us first comment on the main ideas of how to use Theorem \ref{converge} to prove Theorem \ref{windacc} and Theorem \ref{open}.

Concerning the winding number, one of our main motivations comes from \cite{GJYZ}, where they proved that besides convexity and piecewise linearity, $L(E,\mathrm{i}y)$ satisfies Jensen's formula. Thus the zeros in Jensen's formula should be related to the acceleration. It is more transparent by the following well-known Jensen's formula \cite{SS} on the annulus: 
\begin{equation*}
	\begin{split}
		\frac{1}{2\pi}\int_{0}^{2\pi}\log|g(r_{2}\mathrm{e}^{\mathrm{i}x})|\mathrm{d}x&-\frac{1}{2\pi}\int_{0}^{2\pi}\log|g(r_{1}\mathrm{e}^{\mathrm{i}x})|\mathrm{d}x\\
		&= \sum_{r_{1}<|z_{i}|<r_{2}}\log\frac{r_{2}}{|z_{i}|}+\log\frac{r_{2}}{r_{1}}\frac{1}{2\pi\mathrm{i}}\oint_{|z|=r_{1}}\frac{g'(z)}{g(z)} \mathrm{d}z,
	\end{split}
\end{equation*}
where $g$ is holomorphic on $\mathcal{A}(r_{1},r_{2})$ and continuous on $\overline{\mathcal{A}(r_{1},r_{2})}$, and  $\{z_{i}\}$ are the zeros of $g$ such that $r_{1}<|z_{i}|\leq r_{2}$. Take $g=f_{n}(E,x+\mathrm{i}y)$,  then Theorem \ref{windacc} essentially follows from Theorem \ref{converge}. 

Turning to the density of states measure, let us first recall the key point for the proof of the Thouless formula in the self-adjoint case. In the self-adjoint case, Thouless formula follow from the fact that $\mathrm{d}\mathcal{N}_{n}$ converges weakly to a limiting measure $\mathrm{d}\mathcal{N}$ \cite{AS83, CS83a, CS83b}, which depends on functional calculus. However, in the non-self-adjoint case, this method fails due to the non-uniqueness of the spectral measure. To overcome this difficulty, we first prove \eqref{fiLE}, then prove the weak convergence of $\mathrm{d}\mathcal{N}_{n}$ from \eqref{fiLE}. In fact, once \eqref{fiLE} holds, then according to the subharmonicity of logarithmic potential, the convergence of $\mathrm{d}\mathcal{N}_{n}$ is ensured by Widom lemma from potential theory. To apply the Widom lemma, the only technical barrier comes from the convergence of $n^{-1}\log |f_{n}(E,x)|$ for all rationally independent $\alpha\in\mathbb{T}^{d}$ and for a.e. $x\in\mathbb{T}^{d}$ and a.e. $E\in\mathbb{C}$, which is established by using Theorem \ref{converge} and Fubini theorem.


Finally, let us make a comment on Theorem \ref{converge}.  Usually, one can only anticipate \begin{equation}\label{LEdef}
	L(E,\mathrm{i}y)=\lim_{n\rightarrow \infty}\frac{1}{2|n|} \log (|f_{n}(E,x+\mathrm{i}y)|^{2}+|f_{n-1}(E,x+\mathrm{i}y)|^{2}),\ \text{a.e. }x \in\mathbb{T}^{d}.	
\end{equation}
However, it is completely non-trivial to get \eqref{fiLE} from \eqref{LEdef}.  For $v\in C^{\omega}(\mathbb{T}^{d},\mathbb{C})$, Goldstein-Schlag \cite{GS08} proved that \eqref{fiLE} is true for the quasi-periodic shift $x\mapsto x+\alpha$ with $\alpha\in\mathrm{SDC}$. If $v$ is a trigonometric polynomial on $\mathbb{T}^{d}$, they also proved \eqref{fiLE} in \cite{GS08} for the skew shift $(x,y)\mapsto (x+y,y+\alpha)$ for a.e. $\alpha\in\mathbb{T}^{d-1}$. 
Geng-Tao \cite{GT} proved \eqref{fiLE} for Jacobi operators with Brjuno-R\"ussmann frequency. We emphasize that ``a.e.'' in Theorem \ref{converge} can not be improved to ``all'' by a counter-example of Jitomirskaya-Liu \cite{JL08}.  We also mention in the random case, Goldsheid-Khoruzhenko \cite{GK05} proved that if $\{v_{n}\}$ is a sequence of complex i.i.d. random sequence, then
\begin{equation*}
	\lim_{n\rightarrow \infty} \frac{1}{n}\log|f_{n}(E)|=L(E),\quad \text{with probability one}.
\end{equation*}
Concerning the proof of Theorem \ref{converge}, we will give more comments in Section 3.

\section{Preliminary}
Denote by $C^\omega(\mathbb{T}^d,*)$ the
set of all  $*$-valued analytic functions ($*$ will usually denote $\mathbb{C}$, 
$\mathrm{SL}(2,\mathbb{C})$ and etc.). 
For any $F\in C^\omega(\mathbb{T}^d,*)$, we let 
$\displaystyle \langle F\rangle:= \frac{1}{(2\pi)^{d}}\int_{\mathbb{T}^{d}} F(x) \mathrm{d}x$ and define the norm $\|F\|_h:=  \sup_{x\in\mathbb{T}^{d}_{h} } \| F(x)\|$, where $\|\cdot\|$ denotes the usual matrix norm. For any $\mathbf{k}\in\mathbb{Z}^{d}$, we define $|\mathbf{k}|=|\mathbf{k}_{1}|+\cdots+|\mathbf{k}_{d}|$.

\subsection{Quasi-periodic cocycle and uniform hyperbolicity}
	
Given $A \in C^\omega(\mathbb{T}^d, \mathrm{SL}(2, \mathbb{C}))$ and $(1, \alpha)$ rationally independent, one can define the quasi-periodic cocycle $(\alpha, A)$:
\begin{eqnarray*}
	(\alpha, A):\left\{\begin{array}{c}
		\mathbb{T}^d \times \mathbb{C}^2 \rightarrow \mathbb{T}^d \times \mathbb{C}^2, \\
		(x, v) \mapsto(x+\alpha, A(x) \cdot v).
	\end{array}\right.	
\end{eqnarray*}	
We say that $(\alpha,A)$ is uniformly hyperbolic if there exist continuous functions $u,s: \mathbb{T}^{d} \rightarrow \mathbb{PC}^{2}$, called the unstable and stable directions, and $n \geq 1$ such that
\begin{itemize}
	\item[(1)]  for every $x \in \mathbb{T}^{d}$ we have	$A(x)\cdot u(x)=u(x+\alpha)$ and $A(x) \cdot s(x)=s(x+\alpha),$
	\item[(2)] for every unit vector $w \in s(x)$ we have $\|A_n(x) \cdot w\|<1$,
	\item[(3)] for every unit vector $w \in u(x)$ we have $\|A_n(x) \cdot w\|>1$.
	\end{itemize}
It is clear that if $(\alpha,A)$ is uniformly hyperbolic then $L(\alpha,A)>0$. From now on, $(\alpha, A)\in \mathcal{UH}$ means $(\alpha,A)$ is uniformly hyperbolic.
	
\subsection{Schrödinger operators and Schrödinger cocycles}
 We consider the following complex-valued quasi-periodic Schrödinger operators:
\begin{equation*}
	(H(x)\psi)_n=\psi_{n+1}+\psi_{n-1}+v( x+n\alpha) \psi_n, \quad n \in \mathbb{Z}.
\end{equation*}
Given the initial data $(\psi_1,\psi_0)^{\rm T}$, the solution of the difference equation $H(x) \psi=E \psi$ for $n>0$ can be expressed in the form
\begin{equation*}
	\begin{pmatrix}
		\psi_{n+1} \\
		\psi_{n}
	\end{pmatrix}=M_n(E,x)
	\begin{pmatrix}
		\psi_1 \\
		\psi_0
	\end{pmatrix},	
\end{equation*}
where 
\begin{equation*}	
	M_n(E,x)=\prod_{k=n}^1 S_{E,v}(x+k\alpha) \ \text{with} \ S_{E,v}(x)=\begin{pmatrix}
		E-v(x) & -1 \\
		1 & 0
	\end{pmatrix}.
\end{equation*}
We call $(\alpha, S_{E,v})$ the Schrödinger cocycle. The matrix $M_n(x)$ is called the $n$-step transfer matrix. Denote
\begin{equation*}
	f_{n}(E,x):=\det(R_{[1, n]}(H(x)-E) R_{[1, n]}),
\end{equation*}
where $R_{\Lambda}$ denoting the restriction operator to $\Lambda \subset \mathbb{Z}$. Then the $n$-step transfer matrix can be written as	
\begin{equation*}
	M_{n}(E,x)=\begin{pmatrix}
			f_{n}(E,x)&-f_{n-1}(E,x+\alpha)\\
			f_{n-1}(E,x)&-f_{n-2}(E,x+\alpha)
		\end{pmatrix}.
\end{equation*}
The spectrum, denote by $\Sigma_{x}$, of  $H(x)$ is closely related with the dynamical behavior of the Schr\"odinger cocycle  $(\alpha,S_{E,v})$. In the self-adjoint case, i.e. the potential $v$ is real-valued, then by the well-known result of Johnson \cite{John86},  $E\notin \Sigma_{x}$ if and only if $(\alpha,S_{E,v})\in\mathcal{UH}$. The following result extends  Johnson's result \cite{John86} to the non-self-adjoint case.
	\begin{proposition}\label{equi} \cite{GJYZ}
		Suppose that   $v:\mathbb{T}^d\rightarrow \mathbb{C}$ a complex-valued continuous function, then there is some $\Sigma\subset \mathbb{C}$ such that $\Sigma_x=\Sigma$ for all $x\in\mathbb{T}^d$. Moreover,  $E\notin \Sigma$ if and only if $(\alpha,S_{E,v})\in \mathcal{UH}$.
	\end{proposition}

	\subsection{Global theory of one-frequency quasi-periodic cocycles}
	\begin{definition}
		We say that $(\alpha,A)\in \mathbb{R} \backslash \mathbb{Q} \times C^\omega(\mathbb{T},\mathrm{SL}(2,\mathbb{C}))$ is regular if $L(\alpha,A_{\mathrm{i}y})$ is affine for $y$ in a neighborhood of $0$.
	\end{definition}
Let us recall the following two fundamental results in Avila's global theory of one-frequency quasi-periodic $\mathrm{SL}(2, \mathbb{C})$ cocycles.
	\begin{theorem}\cite{Av0}\label{regular}
		Let $(\alpha, A) \in \mathbb{R} \backslash \mathbb{Q} \times C^\omega(\mathbb{T},\mathrm{SL}(2,\mathbb{C}))$. Suppose that $L(E)>0$. Then $L(\alpha, A)$ is regular if and only if $(\alpha, A)\in\mathcal{UH}$.
	\end{theorem}
	
	\begin{theorem}\cite{Av0}\label{cuh}
		For any $(\alpha,A)\in \mathbb{R}\backslash\mathbb{Q}\times C^{\omega}(\mathbb{T},\mathrm{SL}(2,\mathbb{C}))$, there exists $h'>0$ such that either
		\begin{enumerate}[font=\normalfont, label={(\arabic*)}]
			\item $L(\alpha, A_{\mathrm{i}y})=0$ $($and $\omega^{+}(\alpha, A_{\mathrm{i}y})=0)$ for every $0<y<h'$, or
			\item $(\alpha, A_{\mathrm{i}y})$ is uniformly hyperbolic for every $0<y<h'$.
		\end{enumerate}
	\end{theorem}

	\subsection{Some facts on subharmonic functions}	
		Subharmonic functions play a crucial role in potential theory. We collect some of the basic properties of subharmonic functions.
	\begin{lemma}[Riesz decomposition theorem \cite{GS08}]\label{riesz}
		Let $u:\Omega\rightarrow \mathbb{R}$ be a subharmonic function on a domain $\Omega\subset \mathbb{C}$. Suppose that $\partial\Omega$ consists of finitely many piecewise $C^{1}$ curves. There exists a positive measure $\mu$ on $\Omega$ such that for any $\Omega_{1}\Subset \Omega$(i.e. $\Omega_{1}$ is a compactly contained subregion of $\Omega$),
		\begin{equation}\label{decom}
			u(z)=\int_{\Omega_{1}} \log |z-\zeta|\mathrm{d}\mu(\zeta)+g(z),
		\end{equation}
		where $g$ is harmonic on $\Omega_{1}$ and $\mu$ is unique with this property. Moreover, for any $\Omega_{2}\Subset \Omega_{1}$, $\mu$ and $g$ satisfy the bounds
		\begin{equation*}
			\begin{split}
				&\mu(\Omega_{1})\leq C(\Omega, \Omega_{1})\bigg(\sup_{\Omega}u-\sup_{\Omega_{1}}u\bigg),\\
				&\Big\| g-\sup_{\Omega_{1}}u\Big\|_{L^{\infty}(\Omega_{2})} \leq C(\Omega, \Omega_{1},\Omega_{2})\bigg(\sup_{\Omega}u-\sup_{\Omega_{1}}u\bigg).
			\end{split}
		\end{equation*}
	\end{lemma}

	\begin{definition}
		Suppose that $u: \Omega_{1}\times \cdots \times \Omega_{d} \rightarrow \mathbb{R}\cup \{-\infty\}$ is continuous. Then $u$ is said to be separately subharmonic, if for any $1\leq j\leq d$ and $z_{k}\in \Omega_{k}$ for $k\neq j$ the function
		\begin{equation*}
			z\mapsto u(z_{1},\cdots, z_{k-1},z,z_{k+1},\cdots,z_{d})
		\end{equation*}
		is subharmonic in $z\in \Omega_{j}$.
	\end{definition}
	
	We recall the definition of Cartan's sets and the definition is motivated by the following Cartan's estimate of analytic function with several variables.
	\begin{definition}[Cartan's set \cite{GS08}]
		Let $H\gg1$. For any arbitrary subset $\mathcal{B}\subset \mathbb{D}(z_{0},1)\subset \mathbb{C}$, we say $\mathcal{B}\in \mathrm{Car}_{1}(H,J)$ if $\mathcal{B}\subset \cup_{j=1}^{j_{0}} \mathbb{D}(z_{j},r_{j})$ with $j_{0}<J$ and 
		$\sum_{j}r_{j} < \mathrm{e}^{-H}.$
		Here $\mathbb{D}(z,r)$ means the complex disk center at $z$ with radius $r$. If $d$ is a
		positive integer greater than one and $\mathcal{B}\subset \prod_{j=1}^{d} \mathbb{D}(z_{j,0},1) \subset \mathbb{C}^{d}$ then one can define inductively that $\mathcal{B} \in \mathrm{Car}_{d}(H,J)$ if for any $1\leq j\leq d$ there exists $\mathcal{B}_{j} \subset \mathbb{D}(z_{j,0},1)\subset \mathbb{C}$, $\mathcal{B}_{j}\in \mathrm{Car}_{1}(H,J)$ so that $\mathcal{B}_{z}^{(j)}\in \mathrm{Car}_{d-1}(H,J)$ for any $z\in \mathbb{C}\backslash \mathcal{B}_{j}$, here $\mathcal{B}_{z}^{(j)}=\{(z_{1},\cdots,z_{d})\in \mathcal{B}: z_{j}=z\}$.
	\end{definition}

	\begin{lemma}[Cartan's estimate \cite{GS11}]\label{car}
		Let $\varphi(z_{1},\cdots,z_{d})$ be an analytic function defined in a polydisk $\mathcal{P}= \prod_{j=1}^{d} \mathbb{D}(z_{j,0},1)$ with $z_{j,0}\in \mathbb{C}$. Let $M\geq \sup_{z\in \mathcal{P}} \log |\varphi(z)|$, $m\leq \log |\varphi(z_{0})|$, $z_{0}=(z_{1,0},\cdots,z_{d,0})$. Given $H\gg1$, there exists a set $\mathcal{B}\subset \mathcal{P}$, $\mathcal{B}\in \mathrm{Car}_{d}(H^{1/d}, J)$, $J=C_{d}H(M-m)$, such that
		\begin{equation*}
			\log |\varphi(z)|>M-C_{d}H(M-m)
		\end{equation*}
		for any $z\in \prod_{j=1}^{d}\mathbb{D}(z_{j,0},1/6) \backslash \mathcal{B}$.
	\end{lemma}
	\begin{remark}\label{carmark}
		The radius $1/6$ in Lemma \ref{car} was chosen to allow the use of Theorem 4 in \cite{Le} as stated. However, it is straightforward to obtain the following stronger statement: Given $\gamma\in (0,1)$, the lower bound is valid for all $z\in \prod_{j=1}^{d} \mathbb{D}(z_{j,0},1-\gamma)$. The influence of $\gamma$ is only felt in the constants $C_{d}$ and thus $C_{d}$ should be replaced by $C_{d,\gamma}$. 
	\end{remark}
	The definition of the Cartan's sets give implicit information about their measures.
	\begin{lemma}\label{carmes}\cite{GSV19}
		If $\mathcal{B}\in \mathrm{Car}_{d}(H,J)$ then  we have
		$	{\rm mes}(\mathcal{B}\cap \mathbb{R}^{d}) \leq C(d)\mathrm{e}^{-H}.
$
	\end{lemma}

\section{Large deviation theorem}
In this section, we are going to prove a large deviation theorem for the entries of the transfer matrix for the fixed rationally independent frequency. To simplify the notation we suppress the dependence on $E$ sometimes. We define
\begin{equation*}
	L_{n}(E,x+\mathrm{i}y)=\frac{1}{n}\log \|M_{n}(x+\mathrm{i}y)\|
\end{equation*}
and
\begin{equation*}
	L_{n}(E,\mathrm{i}y)=\frac{1}{(2\pi)^{d}}\int_{\mathbb{T}^{d}} L_{n}(E,x+\mathrm{i}y) \mathrm{d}x.
\end{equation*}
In what follows, the dependence on $E,y$ is locally uniform,  and all the constants that depend on $v$ will only depend on the norm of $v$ in various complexified torus. The main result of this section is the following.
\begin{theorem}\label{LDTf}
	Let $d\in \mathbb{N}^{+}$, $h'\in(0,h)$, $v\in C^{\omega}(\mathbb{T}^{d}_{h},\mathbb{C})$, $E\in\mathbb{C}$, $y\in\mathbb{R}^{d}$ with $|y_{j}|<h'(1\leq j\leq d)$. Assume $L(E,\mathrm{i}y)>\gamma>0$.
	Then there exist  $l_{0}=l_{0}(E,v,\gamma,h,h',d)$, $\Cl[sig]{thmsig}=\Cr{thmsig}(d)$, $\Cl{thmC}=\Cr{thmC}(E,v,h,h',d)$, $\Cl[K]{thmK}=\Cr{thmK}(E,v,\gamma,h,h',d)$ such that the following hold.
	\begin{enumerate}[font=\normalfont, label={(\arabic*)}]
		\item\label{lp}
		For any $l\geq l_{0}$ and $1\leq p<\infty$, there exists $\Cl{lp}=\Cr{lp}(E,v,h,h',d,p)$ such that
		\begin{equation*}
			\|\log |f_{l}(\cdot+\mathrm{i}y)|\|_{L^{p}(\mathbb{T}^{d})} \leq \Cr{lp} l.
		\end{equation*}
		\item \label{item:ldt} 
		Let $K>\Cr{thmK}$ and assume $\alpha\in \mathbb{T}^{d}$ with
		\begin{equation*}
			\|\langle {\bf k},\alpha\rangle\|_{\mathbb{T}}>\delta \quad \text{for all} \ 0<|{\bf k}|<K.
		\end{equation*}
	Then for any $n>K\delta^{-1}$,
		\begin{equation}\label{ldtf}
			{\rm mes} \{x\in \mathbb{T}^{d}: |\frac{1}{n} \log |f_{n}(E,x+\mathrm{i}y)|-L_{n}(E,\mathrm{i}y)|>K^{-\Cr{thmsig}}\} < {\rm e}^{-\Cr{thmC}K^{\Cr{thmsig}}}.
		\end{equation}
		\item\label{item:growth} Moreover, if $n>K^2\delta^{-1}$,
		\begin{equation*}\label{average1}
			\bigg|\frac{1}{(2\pi)^d}\int_{\mathbb{T}^{d}}\frac{1}{n} \log |f_{n}(E,x+\mathrm{i}y)| \mathrm{d}x-L(E,\mathrm{i}y)\bigg|<\Cr{thmC}K^{-\Cr{thmsig}}.
		\end{equation*}
	\end{enumerate}
\end{theorem}

As direct consequence of Theorem \ref{LDTf}, we can finish the proof of Theorem \ref{converge}:
\begin{proof}[Proof of Theorem \ref{converge}:]
Fix $E\in \mathbb{C}$ and $y\in\mathbb{R}^{d}$ with $|y_{j}|<h(1\leq j\leq  d)$ such that $L(E,\mathrm{i}y)>0$. 
Denote by $\Theta_{K}$ the set in the R.H.S. of \eqref{ldtf}.  By Theorem \ref{LDTf}\ref{item:ldt}, we have $\sum_{K=1}^{\infty} {\rm mes}(\Theta_{K})<\infty$.
Thus by Borel-Cantelli Lemma, ${\rm mes} (\limsup_{K\rightarrow\infty}  \Theta_{K}) =0$.  Moreover if $n>K\delta^{-1}$, we have $\delta\rightarrow 0$ and $n\rightarrow \infty$ as $K\rightarrow \infty$.  Therefore there exists a full measure set $\Theta=\mathbb{T}^{d}\backslash \limsup \Theta_{K}$ depending on $E$ and $y$ such that for any $x\in \Theta$,
\begin{equation*}
	\lim_{K\rightarrow \infty}\frac{1}{n} \log |f_{n}(E,x+\mathrm{i}y)| = \lim_{n\rightarrow \infty}\frac{1}{n} \log |f_{n}(E,x+\mathrm{i}y)| =L(E,\mathrm{i}y).
\end{equation*}
Combining the above with Theorem \ref{LDTf}\ref{item:growth} proves the Theorem \ref{converge}.\end{proof}

Before giving its full proof of Theorem \ref{LDTf},  we first review previous results.
If $\alpha\in\mathrm{SDC}$, Goldstein-Schlag  \cite{GS08}  proved that $\|\log|f_{l}|\|_{L^{p}} \lesssim l^{1+\varepsilon}$ for any $\varepsilon>0$.  Later, Tao-Voda \cite{TV} improved the bound to  $\|\log|f_{l}|\|_{L^{p}} \lesssim l$. So Theorem \ref{LDTf}\ref{lp} extends this $L^{p}$ estimate to rationally independent frequency. This fact will be crucial to establish the large deviation theorem of the determinant for rationally independent frequency.

Let us recall the large deviation theorem by Bourgain \cite{Bo05} where $\alpha\in \mathbb{T}^{d}$ satisfies a certain restricted Diophantine condition assumed in Theorem \ref{LDTf}\ref{item:ldt}.

\begin{theorem}\cite{Bo05}\label{LDT}
		Let $d\in \mathbb{N}^{+}$, $h>0$, $v\in C^{\omega}(\mathbb{T}^{d}_{h},\mathbb{C})$, $E\in\mathbb{C}$. Assume $\alpha\in \mathbb{T}^{d}$ with
		\begin{equation*}
				\|\langle \mathbf{k},\alpha\rangle\|_{\mathbb{T}} > \delta\quad \text{for all} \ 0<|\mathbf{k}|<K.
			\end{equation*}
		Moreover, suppose  that
		\begin{equation*}
				n>K\delta^{-1}.
			\end{equation*}
	There exist $\Cl[sig]{bousig}=\Cr{bousig}(d)$ and $\Cl{bouC}=\Cr{bouC}(E,v,h,d)$ such that
	\begin{equation*}
		{\rm mes}\{x\in \mathbb{T}^{d}:|L_{n}(E,x)-L_{n}(E)|>\Cr{bouC}K^{-\Cr{bousig}}\} <\mathrm{e}^{-\Cr{bouC}K^{\Cr{bousig}}}.
	\end{equation*}
\end{theorem}

Moreover, the large deviation theorem is valid in the complex region $\mathbb{T}^{d}_{h'}$ via replacing $M_{n}(x)$ by $M_{n}(x+{\rm i}y)$  since its proof only involves the subharmonicity. More precisely, we have the following:
\begin{corollary}\label{ldtexten}
	Under the same assumptions as in Theorem \ref{LDT}, for any $h'\in(0,h)$ there exist $\Cl[sig]{bousigh}=\Cr{bousigh}(d)$ and $\Cl{bouCh}=\Cr{bouCh}(E,v,h,h',d)$ such that for any  $|y|<h'(1\leq j\leq d)$, 
	\begin{equation*}
		{\rm mes} \{x\in \mathbb{T}^{d}: |L_{n}(E,x+\mathrm{i}y)-L_{n}(E,{\rm i}y)|>\Cr{bouCh}K^{-\Cr{bousigh}}\} < {\rm e}^{-\Cr{bouCh}K^{\Cr{bousigh}}}.
	\end{equation*}
\end{corollary}
The large deviation theorem  was first established by Bourgain-Goldstein \cite{BG} for $v\in C^{\omega}(\mathbb{T}^{d},\mathbb{R})$, assuming $\alpha$ is Diophantine.
For any rationally independent frequency,  and $d=1$, the result was previously proved by Bourgain-Jitomirskaya \cite{BJ02}. For $d>1$, Powell \cite{Powell} extended \cite{Bo05} to the analytic $M(2,\mathbb{C})$-valued cocycles. However, giving a large deviation theorem for the first entry instead of the norm of $M_{n}$ is more challenging. This was first done by Goldstein-Schlag \cite{GS08}, who proved that for $\alpha\in \mathrm{SDC}$, in the regime of positive Lyapunov exponent, the large deviations estimate extends to the entries of the transfer matrix, i.e., for any sufficiently large $n$,
\begin{equation*}
	{\rm mes} \{x\in \mathbb{T}^{d}: |\log |f_{n}(x)|-nL_{n}(E)|>n^{-c}\} < {\rm e}^{-n^{c}}.
\end{equation*}
Thus, compare to \cite{Bo05, BJ02, Powell}, Theorem \ref{LDTf}\ref{item:ldt} generalizes the large deviation theorem from the transfer matrix to its first entry. And compare to \cite{GS08}, Theorem \ref{LDTf}\ref{item:ldt} remove the assumption of $\alpha\in \mathrm{SDC}$ so that it works for all $\alpha\in\mathbb{T}^{d}$ rationally independent.

\subsection{Proof of Theorem \ref{LDTf}}

\begin{proof}[Proof of Theorem \ref{LDTf}\ref{lp}]
The main difficulty is estimating the measure of $x\in\mathbb{T}^{d}$ that makes $f_{l}(x+\mathrm{i}y)$ close to zero. Thus the \L{}ojasiewicz-type estimate will help us.

The following Lemma \ref{GS08} proved by Goldstein-Schlag in \cite{GS08} relates the supremum of a separately subharmonic function over $\mathbb{T}^{d}_{h}$ to that over $\mathbb{T}^{d}$, which is used to establish the \L{}ojasiewicz-type estimate, see the following \eqref{cartan}.
	
	\begin{lemma}\cite{GS08}\label{GS08}
		Let $u(z_{1},\cdots,z_{d})$ be a separately subharmonic function on $\mathbb{T}_{h}^{d}$ satisfying $\sup_{z\in \mathbb{T}_{h}^{d}} u(z)\leq \mathcal{M}$.
		There exist $C_{d}, C_{h,d}$ such that, if for some $\varrho\in (0,1)$ and some $\mathcal{S}>0$, 
		\begin{equation*}
			{\rm mes} \{x\in\mathbb{T}^{d}:u(x)<-\mathcal{S}\} > \varrho,
		\end{equation*}
		then we have
		\begin{equation*}
			\sup_{x\in \mathbb{T}^{d}} u(x) \leq C_{h,d} \mathcal{M}-\frac{\mathcal{S}}{C_{h,d}\log^{d}(C_{d}/\varrho)}.
		\end{equation*}
	\end{lemma}
	
	To shortern the notation, we define a constant $C_{\star}=C_{\star}(E,v,h,h',d)$ by
	\begin{equation}\label{cstar}
		\sqrt{C_{\star}}:=\max \{\Cr{bouCh}, 4(d+1)C_{E}^{v},C_{h-h',d}^{2}C_{E}^{v} (\log C_{d})^{d}\},
	\end{equation}
	where $C_{d}$ and $C_{h-h',d}$ are defined in Lemma \ref{GS08}, and 
	\begin{equation*}
		C_{E}^{v}:=\log(|E|+\|v\|_{h}+2).
	\end{equation*}
	It is obviously that for any $E\in\mathbb{C}$,
	\begin{equation}\label{uup}
		 \sup_{x\in\mathbb{T}^{d}}L_n(E,x+\mathrm{i}y)\leq C_{E}^{v}. 	
	 \end{equation}
	Let us show that there exists constant
	$l_{0}$ such for any $K\geq1$ and $l\geq l_{0}$,
	\begin{equation}\label{cartan}
		{\rm mes} \{x\in \mathbb{T}^{d}: |f_{l}(x+\mathrm{i}y)|\leq {\rm e}^{-C_{\star}Kl}\} \leq {\rm e}^{-K ^{\frac{1}{d}}}.
	\end{equation}
	Assume that the statement \eqref{cartan} fails for some $l$ and $K$. Since $u(x):=\log |f_{l}(x+\mathrm{i}y)|$ is a separately subharmonic function on $\mathbb{T}^{d}_{h-h'}$, we apply Lemma \ref{GS08} to $u(x)$ with 
	\begin{equation*}
		\mathcal{M}=C_{E}^{v}l,\quad \mathcal{S}=C_{\star}Kl,\quad \varrho={\rm e}^{-K ^{\frac{1}{d}}}.	
	\end{equation*}
	Then by Lemma \ref{GS08} and the choice of $C_{\star}$ in \eqref{cstar}, one can get that
	\begin{equation*}
		\begin{split}
			\sup_{x\in \mathbb{T}^{d}} u(x) &\leq C_{h-h',d}\mathcal{M}-\frac{\mathcal{S}}{C_{h-h',d}\log^{d}(C_{d}/\varrho)}\\
			&\leq \bigg(C_{h-h',d}C_{E}^{v}-\frac{C_{\star}}{C_{h-h',d}(\log C_{d})^{d}}\bigg) l\\
			& \leq -\sqrt{C_{\star}}l,
		\end{split}
	\end{equation*}
	which means $\sup_{\mathbb{T}^{d}} |f_{l}(x+\mathrm{i}y)| \leq \exp(-\sqrt{C_{\star}}l)$. Recall that $\det M_{l}(x+\mathrm{i}y)=1$ with
	\begin{equation*}
		M_{l}(x+\mathrm{i}y)=\begin{pmatrix}
			f_{l}(x+\mathrm{i}y)&-f_{l-1}(x+\mathrm{i}y+\alpha)\\
			f_{l-1}(x+\mathrm{i}y)&-f_{l-2}(x+\mathrm{i}y+\alpha)
		\end{pmatrix}.
	\end{equation*}
	By \eqref{uup}, we have  $\sup_{x\in\mathbb{T}^{d}}\log|f_{l}(x+\mathrm{i}y)|\leq C_{E}^{v} l$, thus 
	\begin{equation*}
		\sup_{x\in \mathbb{T}^{d}} |f_{l-1}(x+\mathrm{i}y)\cdot f_{l-1}(x+\mathrm{i}y+\alpha)-1|\leq \exp(-\sqrt{C_{\star}}l/2),
	\end{equation*}
	where we use $\sqrt{C_{\star}}\geq 4(d+1)C_{E}^{v}$.
	
	To obtain the contradiction, we use the following result.
	\begin{lemma}\label{closeone}\cite{GS08}
		Suppose that $f\in C^{1}(\mathbb{T}^{d},\mathbb{C})$ satisfies $|f|+|\nabla f|\leq \mathcal{H}$ for some $\mathcal{H}\geq 1$. Let $0<\epsilon<\mathcal{H}^{-2(d+1)}$ and assume that 
		\begin{equation*}
			\sup_{x\in \mathbb{T}^{d}} |f(x)f(x+\alpha)-1|\leq \epsilon.
		\end{equation*}
		Then there exists a constant $\Cl{one}=\Cr{one}(d)$ such that
		\begin{equation*}
			\sup_{x\in \mathbb{T}^{d}}|f(x)^{2}-1|\leq \Cr{one} \Big(\epsilon \mathcal{H}^{2(d+1)}\Big)^{\frac{1}{2+d}}.
		\end{equation*}
	\end{lemma}
	
	Now we apply Lemma \ref{closeone} to $f_{l-1}(\cdot+\mathrm{i}y)$ with 
	\begin{equation*}
		\mathcal{H}= \exp(C_{E}^{v}l),\quad \epsilon=\exp(-\sqrt{C_{\star}}l/2)
	\end{equation*}
	to conclude that $\sup_{x\in \mathbb{T}^{d}} |f_{l-1}(x+\mathrm{i}y)|\leq 2$ for the sufficiently large $l$. Thus the iteration of 
	\begin{equation*}
		f_{l}(x+\mathrm{i}y)=(E-v(x+\mathrm{i}y+l\alpha)) \cdot f_{l-1}(x+\mathrm{i}y)-f_{l-2}(x+\mathrm{i}y)
	\end{equation*}
	implies that $\sup_{x\in \mathbb{T}^{d}} |f_{l-2}(x+\mathrm{i}y)| \leq 1+2(|E|+\|v\|_{h})$, and thus
	\begin{equation*}
		\sup_{x\in \mathbb{T}^{d}} \|M_{l}(x+\mathrm{i}y)\| \leq 6+2(|E|+\|v\|_{h}).
	\end{equation*}
	However, it contradicts to the assumption that $L(E,\mathrm{i}y)>\gamma$. This proves \eqref{cartan}.

	With the help of \eqref{cartan} one can finish the  proof. We decompose $\mathbb{T}^{d}=(\cup_{K\geq 0} \mathcal{Q}_{K}(y))\cup\mathcal{Q}(y)$, where
	\begin{equation*}
		\begin{split}
			&\mathcal{Q}(y) = \{x\in \mathbb{T}^{d}:|f_{l}(x+\mathrm{i}y)|> 1\},\\
			&\mathcal{Q}_{K}(y)=\{x\in\mathbb{T}^{d}: {\rm e}^{-C_{\star}(K+1)l} < |f_{l}(x+\mathrm{i}y)|\leq {\rm e}^{-C_{\star}Kl}\}.
		\end{split}
	\end{equation*}
	If $x\in \mathcal{Q}(y)$, then
	\begin{equation}\label{positive}
		\biggl(\frac{1}{(2\pi)^d}\int_{\mathcal{Q}} \Big|\log|f_{l}(x+\mathrm{i}y)|\Big|^{p} \mathrm{d}x \biggr)^{1/p}\leq \log \|f_{l}(\cdot+\mathrm{i}y)\|_{\infty}\leq C_{E}^{v}l.
	\end{equation}
	If $x\in \mathcal{Q}_{K}(y)$ with $K\geq 1$, by \eqref{cartan}, we have ${\rm mes}(\mathcal{Q}_{K}(y)) \leq {\rm e}^{-K ^{\frac{1}{d}}}$.
	By calculation,
	\begin{equation}\label{negative}
		\sum_{K\geq 0} \frac{1}{(2\pi)^d}\int_{\mathcal{Q}_{K}}\Big|\log |f_{l}(x+\mathrm{i}y)|\Big|^{p} \mathrm{d}x \leq (C_{\star}l)^p+\sum_{ K> 0} {\rm e}^{-K^{\frac{1}{d}}} (C_{\star}(K+1)l)^{p}.
	\end{equation}	 
Hence we finish the proof of Theorem \ref{LDTf}\ref{lp} by combining  $(\ref{positive})$ with $(\ref{negative})$ and letting
	\begin{equation*}
		\Cr{lp}:=C_E^v+C_{\star}+C_{\star}\big(\sum_{ K> 0} {\rm e}^{-K^{\frac{1}{d}}}(K+1)^{p}\big)^{\frac{1}{p}}.
	\end{equation*}
\end{proof}

\begin{proof}[Proof of Theorem \ref{LDTf}\ref{item:ldt}]
To obtain Theorem \ref{LDTf}\ref{item:ldt}, we need to estimate the uniform upper bound and the average lower bound of $\log |f_{n}(\cdot+\mathrm{i}y)|$ respectively. The following Proposition \ref{upper} and Proposition \ref{lower} provide what we need and we will first use these two propositions, and postpone the proofs to subsection 3.2 and subsection 3.3.

	\begin{proposition}\label{upper}
		Under the same assumptions as in Theorem \ref{LDT}, let $h'\in(0,h)$ and $y\in\mathbb{R}^d$ with $|y_j|<h'(1\leq j\leq d)$.
		There exist constants $\Cl[sig]{uppersig}=\Cr{uppersig}(d)$, $\Cl{upperC}=\Cr{upperC}(E,v,h,h',d)$, $\Cl[K]{upperK}=\Cr{upperK}(E,v,h,h',d)$ such that for any $K>\Cr{upperK}$, we have
		\begin{equation*}
			\sup_{x\in \mathbb{T}^{d}}\frac{1}{n} \log \|M_{n}(x+\mathrm{i}y)\| \leq L_{n}(E,\mathrm{i}y) + \Cr{upperC}K^{-\Cr{uppersig}}.
		\end{equation*}
	
	\end{proposition}

	\begin{proposition}\label{lower}
			Under the same assumptions as in Theorem \ref{LDT}, let $h'\in(0,h)$ and $y\in\mathbb{R}^{d}$ with $|y_{j}|<h'(1\leq j\leq d)$. Assume $L(E,\mathrm{i}y)>\gamma>0$.
			 There exist constants $\Cl[sig]{lowersig}=\Cr{lowersig}(d)$, $\Cl{lowerC}=\Cr{lowerC}(E,v,h,h',d)$, $\Cl[K]{lowerK}=\Cr{lowerK}(E,v,\gamma,h,h',d)$ such that 	for any $K> \Cr{lowerK}$, we have
		\begin{equation*}
			\frac{1}{(2\pi)^d}\int_{\mathbb{T}^{d}}\frac{1}{n} \log |f_{n}(x+\mathrm{i}y)| \mathrm{d}x>L_{n}(E,\mathrm{i}y)-\Cr{lowerC}K^{-\Cr{lowersig}}.
		\end{equation*}

	\end{proposition}
	
	Now we prove Theorem \ref{LDTf}\ref{item:ldt}. For simplicity, we shall set $d=2$. Fix $|y_{1}|<h',|y_{2}|<h'$ and let $\displaystyle u(x_1,x_2)=\frac{1}{n}\log |f_{n}(x_1+\mathrm{i}y_1,x_2+\mathrm{i}y_2)|$. 
	It follows from Proposition \ref{upper} and Proposition \ref{lower} that for any $K>\max\{\Cr{upperK},\Cr{lowerK}\}$,
	\begin{equation*}
		\begin{cases}
			\langle u\rangle >L_{n}(E,\mathrm{i}y)-\Cr{lowerC}K^{-\Cr{lowersig}}, \\ 
			\sup_{\mathbb{T}^{2}} u<L_{n}(E,\mathrm{i}y) +\Cr{upperC}K^{-\Cr{uppersig}}.
		\end{cases}
	\end{equation*}
	To use tools from one-variable subharmonic functions, we let $\displaystyle v(x_1)=\frac{1}{2\pi}\int_{\mathbb{T}} u(x_1,x_2) \mathrm{d}x_2$
	and choose $\Cl[sig]{minsig}=\min\{\Cr{uppersig},\Cr{lowersig}\}$. Then we have
	\begin{equation*}
		\begin{cases}
			\langle v\rangle >L_{n}(E,\mathrm{i}y)-\Cr{lowerC}K^{-\Cr{minsig}}, \\ 
			\sup_{\mathbb{T}} v<L_{n}(E,\mathrm{i}y) +\Cr{upperC}K^{-\Cr{minsig}}.
		\end{cases}
	\end{equation*}
	Denote $\mathcal{P}=\{x_{1}\in \mathbb{T} :v(x_{1})> \langle v\rangle\}$. Then by using 
		$\displaystyle\int_{\mathbb{T}} (v-\langle v\rangle)\mathrm{d}x_{1}=0$, we have
	\begin{equation}\label{L1}
		\begin{split}
			\|v-\langle v\rangle\|_{L^{1}(\mathbb{T})}&=\frac{1}{2\pi}\int_{\mathcal{P}}|v-\langle v\rangle|\mathrm{d}x_{1}  +\frac{1}{2\pi}\int_{\mathbb{T}\backslash \mathcal{P}} |v-\langle v\rangle|\mathrm{d}x_{1}  \\
			&\leq\int_{\mathcal{P}}|v-\langle v\rangle|\mathrm{d}x_{1} \leq \sup_{x_{1}\in \mathbb{T}} v(x_{1})-\langle v\rangle<(\Cr{upperC}+\Cr{lowerC})K^{-\Cr{minsig}}.
		\end{split}
	\end{equation}
	
	To apply the John-Nirenberg inequality, we need to use the following lemma to control the $\mathrm{BMO}$ norm of subharmonic functions.
	
	\begin{lemma}\label{split}\cite{BoGS}
		Suppose that $\varphi$ is subharmonic on $\mathbb{T}_{h}$ and has Riesz decomposition as \eqref{decom} on $\mathbb{T}_{h/2}$ with 
		\begin{equation}\label{rie}
			\mu(\mathbb{T}_{h/2})+\|g\|_{L^{\infty}(\mathbb{T}_{h/4})}<\infty.
		\end{equation}
		Then
		\begin{equation*}
			\|\varphi\|_{\mathrm{BMO}(\mathbb{T})} \leq C_{h} \|\varphi-\langle \varphi\rangle\|^{1/2}_{L^{1}(\mathbb{T})},
		\end{equation*}
		where	\begin{equation*}
			\|\varphi\|_{\mathrm{BMO}(\mathbb{T})}:=\sup_{I\subset \mathbb{T}} \frac{1}{|I|} \int_{I} |\varphi-\langle \varphi\rangle_{I}| \mathrm{d}x, \ \text{with} \ \langle \varphi\rangle_{I}=\frac{1}{|I|} \int_{I}\varphi(x)\mathrm{d}x.
		\end{equation*}
	\end{lemma}
	\begin{remark}
		Lemma \ref{split} was proved essentially in \cite{BoGS} named by ``$\mathrm{BMO}$ splitting lemma'' under the assumption of boundedness of the subharmonic function. However, all that was used in the proof of that lemma was assuming \eqref{rie}.
	\end{remark}

	Now we check $v$ satisfying the conditions in Lemma \ref{split}. Since $v$ is subharmonic on $\mathbb{T}_{h-h'}$, by Lemma \ref{riesz}, one can decompose 
	\begin{equation*}
		v(z)=\int_{\mathbb{T}_{(h-h')/2}} \log |z-\zeta|\mathrm{d}\mu(\zeta) +g(z),
	\end{equation*}
	with the estimate
	\begin{equation*}
		\mu(\mathbb{T}_{(h-h')/2})+\|g\|_{L^{\infty}(\mathbb{T}_{(h-h')/4})} \leq \Cl{riesz}(h,h')(2\sup_{z\in\mathbb{T}_{h-h'}} v(z)-\langle v\rangle).
	\end{equation*}
	Hence by Lemma \ref{split} and \eqref{L1}, there exists constant $\Cl{bmo}=\Cr{bmo}(E,v,h,h',d)$ such that for any $K>\max\{\Cr{upperK}, \Cr{lowerK}\}$,
	\begin{equation}\label{vbmo}
		\|v\|_{\mathrm{BMO}(\mathbb{T})} \leq C_{h-h'} \|v-\langle v\rangle\|^{1/2}_{L^{1}(\mathbb{T})}\leq \Cr{bmo} K^{-\frac{\Cr{minsig}}{2}}.
	\end{equation}
	Recall the John-Nirenberg inequality that for any $f\in \mathrm{BMO}(\mathbb{T})$ there exist constants $\check{C},\check{c}>0$ (independent of $f$) such that
	\begin{equation*}
		{\rm mes} \{x\in \mathbb{T}: |f-\langle f\rangle|>\varepsilon\}<\check{C}\mathrm{e}^{-\check{c}\frac{\varepsilon}{\|f\|_{\mathrm{BMO}}}}.
	\end{equation*}
	Combining \eqref{vbmo} with John-Nirenberg inequality gives
	\begin{equation}\label{vldt}
		{\rm mes}\{x\in\mathbb{T}: |v-\langle v\rangle|>K^{-\frac{\Cr{minsig}}{4}}\}<\check{C}\mathrm{e}^{-\Cl{JN}K^{\frac{\Cr{minsig}}{4}}}
	\end{equation}
	for some constant $\Cr{JN}=\Cr{JN}(E,v,h,h',d)$.
	
	Let
$
		\mathcal{V}=\{x_1\in \mathbb{T}: v(x_1)>L_{n}(E,\mathrm{i}y)-K^{-\frac{{\Cr{minsig}}}{4}}\}.
$
	Then ${\rm mes}(\mathbb{T}\backslash \mathcal{V})< \check{C}\exp({-\Cr{JN}K^{\frac{\Cr{minsig}}{4}}})$ according to \eqref{vldt}.
	To study the second variable of $u(x_1,x_2)$, we fix some $x_1\in \mathcal{V}$ and let $w(x_2)=u(x_1,x_2)$. Then there exists $\Cl[K]{wK}=\Cr{wK}(E,v,\gamma,h,h',d)$ such that for any $K>\Cr{wK}$,
	\begin{equation*}
		\begin{cases}
			\langle w\rangle >L_{n}(E,\mathrm{i}y)-K^{-\frac{\Cr{minsig}}{4}}, \\ 
			\sup_{\mathbb{T}} w<L_{n}(E,\mathrm{i}y) +\Cr{upperC}K^{-\Cr{minsig}}<L_{n}(E,\mathrm{i}y) +K^{-\frac{\Cr{minsig}}{4}}.
		\end{cases}
	\end{equation*}
	As similar as \eqref{L1}, one can use Lemma \ref{riesz}, Lemma \ref{split} and John-Nirenberg inequality again to get that for any $x_1\in \mathcal{V}$, 
	\begin{equation*}
		{\rm mes}\{x_2\in \mathbb{T}: |w-\langle w\rangle|>K^{-\frac{\Cr{minsig}}{16}}\} <\check{C}\mathrm{e}^{-\Cl{JN2}K^{\frac{\Cr{minsig}}{16}}}
	\end{equation*}
	for some constant $\Cr{JN2}=\Cr{JN2}(E,v,h,h',d)$.
	
Denote $\mathcal{R}_{x_1}=\{x_2\in \mathbb{T}: |w-\langle w\rangle|> K^{-\frac{\Cr{minsig}}{16}}\}$ and let $	\Cl[K]{GK}=\max\{\Cr{upperK}, \Cr{lowerK}, \Cr{wK}\}, \Cl[sig]{mainsig}=\frac{\Cr{minsig}}{16}$. Denote by $\Theta_{K}$ the set in \eqref{ldtf}. Then by Fubini theorem, there exists $\Cl{mainC}=\Cr{mainC}(E,v,h,h',d)$ such that
\begin{equation*}
	{\rm mes}(\Theta_{K})= \frac{1}{(2\pi)^2}\bigg(\int_{\mathbb{\mathcal{V}}} {\rm mes}(\mathcal{R}_{x_1})\mathrm{d}x_1+\int_{\mathbb{T}\backslash \mathcal{V}} {\rm mes} (\mathcal{R}_{x_1})\mathrm{d}x_1\bigg) <\mathrm{e}^{-\Cr{mainC}K^{\Cr{mainsig}}}.
\end{equation*}
For the general case of $d\geq 2$, one can choose $\Cr{mainsig}=4^{-d}\Cr{minsig}$. This proves Theorem \ref{LDTf}\ref{item:ldt}.
\end{proof}

	\begin{proof}[Proof of Theorem \ref{LDTf}\ref{item:growth}]

We will use the convergence result of Lyapunov exponent \cite{Bo05}.

\begin{proposition}[\cite{Bo05}]\label{continuous}
	Let $d\in \mathbb{N}^{+}$, $v\in C^{\omega}(\mathbb{T}^{d}_{h},\mathbb{C})$, $E\in\mathbb{C}$. Assume $\alpha\in \mathbb{T}^{d}$ with
	\begin{equation*}
		\|\langle {\bf k},\alpha\rangle\|_{\mathbb{T}}>\delta
	\end{equation*}
	for all $0<|{\bf k}|<K$. Moreover, suppose $
		n>K^{2}\delta^{-1}. $ 
 $h>K^{-\Cl[sig]{exten}}$, where $\Cr{exten}=\Cr{exten}(d)$. Then
	\begin{equation*}
		|L(E)-L_{n}(E)|<K^{-\Cr{exten}}.
	\end{equation*}
\end{proposition}
It follows from Proposition \ref{upper} and Proposition \ref{lower} that for any $K>\Cr{GK}$,
\begin{equation*}
	-\Cr{lowerC}K^{-\Cr{minsig}}\leq \frac{1}{(2\pi)^d}\int_{\mathbb{T}^{d}}\frac{1}{n} \log |f_{n}(x+\mathrm{i}y)| \mathrm{d}x -L_{n}(E,\mathrm{i}y)\leq \Cr{upperC}K^{-\Cr{minsig}}.
\end{equation*}
Let $\Cr{thmC}=2\max\{\Cr{upperC},\Cr{lowerC},\Cr{mainC}\}$ and $\Cr{thmsig}=\min\{\frac{\Cr{mainsig}}{2},\Cr{exten}\}$. We choose  $\Cr{thmK}>\Cr{GK}$  sufficiently large  such that for $K>\Cr{thmK}$,
\begin{equation*}
	h-h'>K^{-\Cr{exten}} \ \text{and}\ \ \mathrm{e}^{-\Cr{mainC}K^{\Cr{mainsig}}} <\mathrm{e}^{-\Cr{thmC}K^{\Cr{thmsig}}}.
\end{equation*}
Then by Proposition \ref{continuous},
\begin{equation*}
		\bigg|\frac{1}{(2\pi)^d}\int_{\mathbb{T}^{d}}\frac{1}{n} \log |f_{n}(x+\mathrm{i}y)| \mathrm{d}x-L(E,\mathrm{i}y)\bigg|
		\leq\max\{\Cr{upperC},\Cr{lowerC}\}K^{-\Cr{minsig}}+K^{-\Cr{exten}}\leq\Cr{thmC} K^{-\Cr{thmsig}}.
\end{equation*}
This proves Theorem \ref{LDTf}\ref{item:growth}.
\end{proof}

\subsection{Proof of Proposition \ref{upper}}
The main idea of the proof is using the large deviation theorem of the transfer matrix and the submean value property of separately subharmonic function on $\mathbb{T}^{d}_{h}$ to give the explicit estimate for the dynamics of the base, so that we can get the uniform upper bound of the $\log\|M_{n}(x+\mathrm{i}y)\|$, so does $\log|f_{n}(x+\mathrm{i}y)|$. 

To obtain the large deviation theorem for  $|L_{n}(x+\mathrm{i}y)-L_{n}|$, we need to compare $L_{n}(\mathrm{i}y)$ and $L_{n}$. The following lemma claims that $L_{n}({\rm i}y)$ is in fact Lipschitz continuous with respect to $y$ provided that $|y|$ is sufficiently small. 
\begin{lemma}\cite{GSV16}\label{lip}
		Let $d\in \mathbb{N}^{+}$, $v\in C^{\omega}(\mathbb{T}^{d}_{h},\mathbb{C})$, $E\in\mathbb{C}$. Assume $\alpha\in \mathbb{T}^{d}$, then there exists $\Cl{lip}=\Cr{lip}(E,v,h)$ such that 
	\begin{equation*}
		|L_{n}(E,{\rm i}y)-L_{n}(E)|\leq \Cr{lip} \sum_{j=1}^{d} |y_{j}|, \quad \text{for any}\  |y_{j}|\leq h.
	\end{equation*}
\end{lemma}
Now we finish the proof of Proposition \ref{upper}. By Corollary \ref{ldtexten} and Lemma \ref{lip}, we have
\begin{equation}\label{ly}
	{\rm mes} \{x\in \mathbb{T}^{d}: |\frac{1}{n} \log \|M_{n}(x+{\rm i}y+\mathrm{i}\tilde{y})\|-L_{n}(\mathrm{i}y)|>2\Cr{bouCh}K^{-\Cr{bousigh}}\} < {\rm e}^{- \Cr{bouCh}K^{\Cr{bousigh}}}
\end{equation}
provided that for every $1\leq j\leq d$,
\begin{equation*}
	|\tilde{y}_{j}|\leq r:= \frac{\Cr{bouCh}d^{-1}K^{-\Cr{bousigh}} }{\Cr{lip}(E,v,h-h')}.
\end{equation*}
Due to the submean value property for the separately subharmonic function, for any $x=(x_{1},\cdots,x_{d})$, the following holds
\begin{equation*}
	\log \|M_{n}(x+\mathrm{i}y)\|\leq \frac{1}{(\pi r^{2})^{d}} \int_{\mathbb{D}^{d}} \log \|M_{n}(\xi+\mathrm{i}y+{\rm i}\tilde{y})\| \mathrm{d}\xi \mathrm{d}\tilde{y},
\end{equation*}
where $\mathbb{D}^{d} = \prod_{j=1}^{d} \mathbb{D}(x_{j},r)$. Denote by $\mathcal{B}_{\tilde{y}}\subset \mathbb{T}^{d}$ the set in \eqref{ly}. Let
\begin{equation*}
	\mathcal{B}= \{(\xi,\tilde{y})\in \mathbb{T}^{d}\times (-r,r)^{d}: \xi \in \mathcal{B}_{\tilde{y}}\}.
\end{equation*}
It follows from (\ref{ly}) that ${\rm mes} (\mathcal{B}) \leq {\rm mes}(\mathcal{B}_{\tilde{y}}) \leq {\rm e}^{-\Cr{bouCh} K^{\Cr{bousigh}}}$. On the one hand, 
\begin{equation}\label{hand1}
	\frac{1}{(\pi r^{2})^{d}} \int_{\mathbb{D}^{d}\backslash \mathcal{B}} \frac{1}{n}\log\|M_{n}(\xi+\mathrm{i}y+{\rm i}\tilde{y})\| \mathrm{d}\xi \mathrm{d}\tilde{y} \leq  L_{n}(\mathrm{i}y) +2\Cr{bouCh}K^{-\Cr{bousigh}}. 
\end{equation}
On the other hand, by the definition of $C_{E}^{v}$, for sufficiently large $K$,
\begin{equation}\label{hand2}
	\begin{split}
		\frac{1}{(\pi r^{2})^{d}} \int_{\mathbb{D}^{d}\cap \mathcal{B}} \frac{1}{n}\log \|M_{n}(\xi+\mathrm{i}y+{\rm i}\tilde{y})\| \mathrm{d}\xi \mathrm{d}\tilde{y} &\leq  C_{E}^{v}(\pi r^{2})^{-d} {\rm mes}(\mathcal{B})\\
		& \leq \mathrm{e}^{-\frac{\Cr{bouCh}}{2}K^{{\Cr{bousigh}}}}.
	\end{split} 
\end{equation}
The proof is completed by \eqref{hand1} and \eqref{hand2} and letting $\Cr{upperC}=3\Cr{bouCh}, \Cr{uppersig}=\Cr{bousigh}$.\qed

\subsection{Proof of Proposition \ref{lower}}
Now we shall study the average lower bound of $\log|f_{n}(\cdot+\mathrm{i}y)|$. The following lemma provides the estimate over three determinants, and its proof is slightly modified versions of \cite{GS08}. We give the proof in Appendix \ref{app}. In what follows we will show how to use it to obtain the average lower bound for a single determinant.

\begin{lemma}\label{three}
	Suppose that $l_{0}\leq j_{1}\leq j_{1}+l_{0}\leq j_{2}\leq K^{\frac{1-\Cr{bousigh}}{2}}$. Let
	\begin{equation*}
		\begin{split}
			\mathcal{F}_{K}:=\{x\in\mathbb{T}^{d}:|f_{n}(x+\mathrm{i}y)|+|f_{n}(T^{j_{1}}x+\mathrm{i}y)| +|f_{n}(T^{j_{2}}x+\mathrm{i}y)|&\\
			\leq {\rm e}^{nL_{n}(\mathrm{i}y)-5C_{\star}nK^{-\Cr{bousigh}}}& \}.
		\end{split}
	\end{equation*}
	Then there exist $\Cl[sig]{trisig}=\Cr{trisig}(d)$ and $\Cl[K]{triK}=\Cr{triK}(E,v,\gamma,h,h',d)$ such that for any $K\geq \Cr{triK}$,
	\begin{equation*}
		{\rm mes} (\mathcal{F}_{K})\leq {\rm e}^{-K^{\Cr{trisig}}}.
	\end{equation*}
\end{lemma}

We are going to finish the proof of Proposition \ref{lower}. Let $\mathcal{G}_{K}$ be the set of $x\in \mathbb{T}^{d}$ such that for any $|j|\leq K^{\frac{1-\Cr{bousigh}}{2}}$,
\begin{equation*}
	|L_{n}(x+\mathrm{i}y+j\alpha)-L_{n}(\mathrm{i}y)|<\Cr{bouCh}K^{-\Cr{bousigh}}.
\end{equation*}
By Corollary \ref{ldtexten} we have
$	{\rm mes}(\mathbb{T}^{d}\backslash \mathcal{G}_{K}) \leq K\mathrm{e}^{-\Cr{bouCh}K^{\Cr{bousigh}}}$. For any $j_{1},j_{2},j_{3}$ with
\begin{equation*}
	0\leq j_{1}\leq j_{1}+l_{0}\leq j_{2}\leq j_{2}+l_{0}\leq j_{3} \leq K^{\frac{1-\Cr{bousigh}}{2}},
\end{equation*}	
we define  $\Omega_{K}$ be the set of $x\in \mathcal{G}_{K}$ such that		
\begin{equation*}
	\begin{split}
		\min_{j_{1},j_{2},j_{3}} \bigg(|f_{n}(T^{j_{1}}x+\mathrm{i}y)|+|f_{n}(T^{j_{2}}x+\mathrm{i}y)| &+|f_{n}(T^{j_{3}}x+\mathrm{i}y)|\bigg) \\
		&>{\rm e}^{nL_{n}(\mathrm{i}y)-5C_{\star}nK^{-\Cr{bousigh}}}.
	\end{split}
\end{equation*}
Then by Lemma \ref{three}, for $K>\Cr{triK}$, 
\begin{equation}\label{omegac}
	{\rm mes} (\mathbb{T}^{d}\backslash \Omega_{K})\leq K{\rm e}^{-K^{\Cr{trisig}}}.
\end{equation}

Let $u(x)=\frac{1}{n}\log|f_{n}(x+\mathrm{i}y)|$ and choose $R=[K^{\frac{1-\Cr{bousigh}}{2}} /l_{0}]$. To apply Cartan's estimate to $f_{n}(\cdot+\mathrm{i}y)$, we choose $x_{0}\in \mathbb{T}^{d}$ such that $\langle u\rangle n\leq \log |f_{n}(x_{0}+\mathrm{i}y)|$. By Lemma \ref{car}, Remark \ref{carmark} and Lemma \ref{carmes}, we can choose
\begin{equation*}
	M=C_{E}^{v} n,\quad m=\langle u\rangle n, \quad H=K^{\frac{1-\Cr{bousigh}}{4}}.
\end{equation*}
Then there exists $\Cl{car}=\Cr{car}(E,v,h,h',d)$ such that 
\begin{equation}\label{caru}
	\inf_{1\leq k\leq R}u(T^{kl_{0}}x) >C_{E}^{v}-C_{d,h-h'}K^{\frac{1-\Cr{bousigh}}{4}}(C_{E}^{v}-\langle u\rangle) >-\Cr{car}K^{\frac{1-\Cr{bousigh}}{4}}
\end{equation}
except a set $\mathcal{R}_{K}\subset \mathbb{T}^{d}$ of measure 
\begin{equation*}
	{\rm mes}(\mathcal{R}_{K}) \leq C(d)K^{\frac{1-\Cr{bousigh}}{2}} \mathrm{e}^{-K^{\frac{1-\Cr{bousigh}}{4d}}}.
\end{equation*}

Rewrite
\begin{equation*}
	\langle u\rangle=\frac{1}{(2\pi)^d}\bigg(\int_{\Omega_{K}} +\int_{\mathbb{T}^{d}\backslash \Omega_{K}}\bigg) \frac{1}{R} \sum_{k=1}^{R} u(T^{kl_{0}}x)\mathrm{d}x=:(I)+(II).
\end{equation*}
By the definition of $\Omega_{K}$, there exist at most two terms of $u(T^{kl_{0}}x)$ for $1\leq k\leq R$ that does not obey $u(T^{kl_{0}}x)>L_{n}(\mathrm{i}y)-5C_{\star}K^{-\Cr{bousigh}}$, thus
\begin{equation*}
	\begin{split}
		(I)&\geq  \frac{1}{(2\pi)^d}\int_{\Omega_{K}}  \frac{R-2}{R}(L_{n}(\mathrm{i}y)-5C_{\star}K^{-\Cr{bousigh}})\mathrm{d}x +\frac{1}{(2\pi)^d}\int_{\Omega_{K}} \frac{2}{R}\inf_{1\leq k\leq R}u(T^{kl_{0}}x) \mathrm{d}x\\
		&=:(III)+(IV).
	\end{split}
\end{equation*}
By \eqref{omegac} and the choice of $R$, if we choose $\Cl[sig]{cal}=\min\{\Cr{bousigh}, \frac{1-\Cr{bousigh}}{2}, \frac{\Cr{trisig}}{2}\}$, then there exists $\Cl{iii}=\Cr{iii}(E,v,h,h',d)$ such that for any sufficiently large $K$,
\begin{equation*}
	\begin{split}
		(III)&\geq L_{n}(\mathrm{i}y)-5C_{\star}K^{-\Cr{bousigh}}-\frac{2C_{E}^{v}}{R}- \mathrm{mes}(\mathbb{T}^{d}\backslash \Omega_{K})(2C_{E}^{v}+5C_{\star}K^{-\Cr{bousigh}})\\
		&\geq L_{n}(\mathrm{i}y)-\Cr{iii}K^{-\Cr{cal}}.
	\end{split}
\end{equation*}
By \eqref{caru} and H\"older inequality, there exists $\Cl{iv}=\Cr{iv}(E,v,h,h',d)$ such that for sufficiently large $K$,
\begin{equation*}
	\begin{split}
		(IV) &>\frac{1}{(2\pi)^d}\frac{2}{R}\int_{\Omega_{K}\backslash\mathcal{R}_{K}} -\Cr{car} K^{\frac{1-\Cr{bousigh}}{4}}\mathrm{d}x +\frac{1}{(2\pi)^d}\frac{2}{R}\int_{\Omega_{K}\cap \mathcal{R}_{K}}  \inf_{1\leq k\leq R} u(T^{k l_{0}} x) \mathrm{d}x\\
		&> -\frac{2\Cr{car}}{R}K^{\frac{1-\Cr{bousigh}}{4}}-\frac{1}{(2\pi)^d}\frac{2}{R}\sum_{k=1}^{R} \int_{\Omega_{K}\cap \mathcal{R}_{K}} |u(T^{kl_{0}}x)|\mathrm{d}x\\
		&\geq -2\Cr{car}l_{0}K^{-\frac{1-\Cr{bousigh}}{4}}-2 \mathrm{mes}(\Omega_{K}\cap\mathcal{R}_{K})^{\frac{1}{2}}\|u\|_{L^{2}(\mathbb{T}^{d})}\\
		&\geq -\Cr{iv}K^{-\frac{1-\Cr{bousigh}}{8d}},
	\end{split}
\end{equation*}	
where we use $\|u\|_{L^{2}(\mathbb{T}^{d})}\leq \Cr{lp}(E,v,h,h',d,2)$ via Theorem \ref{LDTf}\ref{lp}. Similarly, by H\"older inequality and Theorem \ref{LDTf}\ref{lp} again, there exists $\Cl{ii}=\Cr{ii}(E,v,h,h',d)$ such that for sufficiently large $K$,
\begin{equation*}
	(II)\geq - \mathrm{mes}(\mathbb{T}^{d}\backslash \Omega_{K})^{\frac{1}{2}}\|u\|_{L^{2}(\mathbb{T}^{d})} \geq -\Cr{ii}\sqrt{K}{\rm e}^{-\frac{1}{2}K^{\Cr{trisig}}}.
\end{equation*}	
This finishes the proof by letting $\Cr{lowersig}=\min\{\frac{1-\Cr{bousigh}}{8d}, \frac{\Cr{cal}}{2}\}$ and combining  $(I)-(IV)$.\qed

\section{Winding number}
Note that the winding number is closely related to the number of zeros of the holomorphic function on the annulus, as explained in the Introduction, to prove the existence of the winding number, the key is to establish Jensen's formula on the annulus. However, noting the determinant $\det(H_{n}(z)-E)$ could be zero on the boundary of the annulus, this kind of Jensen's formula must include the effect from the boundary.
\subsection{Jensen's formula on the closed annulus}
\begin{proposition}[Jensen's formula]\label{jensen}
	Let $f$ be holomorphic in a neighborhood of $\overline{\mathcal{A}(r_{1},r_{2})}$, and let $\{z_{i}\}$ be the zeros of $f$ (counted with multiplicities) in $\overline{\mathcal{A}(r_{1},r_{2})}$. Then
	\begin{equation}\label{jens1}
		\begin{split}
			\frac{1}{2\pi}\int_{0}^{2\pi}\log&|f(r_{2}\mathrm{e}^{\mathrm{i}x})|\mathrm{d}x-\frac{1}{2\pi}\int_{0}^{2\pi}\log|f(r_{1}\mathrm{e}^{\mathrm{i}x})|\mathrm{d}x\\
			&= \sum_{r_{1}<|z_{i}|<r_{2}}\log\frac{r_{2}}{|z_{i}|}+\frac{1}{2}\sum_{|z_{i}|=r_{1}}\log\frac{r_{2}}{r_{1}}+\log\frac{r_{2}}{r_{1}}\frac{1}{2\pi\mathrm{i}}\oint_{|z|=r_{1}}\frac{f'(z)}{f(z)} \mathrm{d}z.
		\end{split}
	\end{equation}
	As a consequence, we have
		\begin{equation}\label{jens2}
		\begin{split}
			\log&\frac{r_{2}}{r_{1}}\frac{1}{2\pi\mathrm{i}}\oint_{|z|=r_{1}}\frac{f'(z)}{f(z)} \mathrm{d}z\\
			&\leq	\frac{1}{2\pi}\int_{0}^{2\pi}\bigg(\log|f(r_{2}\mathrm{e}^{\mathrm{i}x})|-\log|f(r_{1}\mathrm{e}^{\mathrm{i}x})|\bigg)\mathrm{d}x \leq 
			\log\frac{r_{2}}{r_{1}}\frac{1}{2\pi\mathrm{i}}\oint_{|z|=r_{2}}\frac{f'(z)}{f(z)} \mathrm{d}z.		
		\end{split}
	\end{equation}	
\end{proposition}
\begin{proof}
	Let us first recall the following  Jensen's formula on the semi-open annulus \cite{SS}:
	
	\begin{lemma}[\cite{SS}]\label{spencer}
		Let $g$ be holomorphic on $\mathcal{A}(r_{1},r_{2})$ and continuous on $\overline{\mathcal{A}(r_{1},r_{2})}$, and let $\{z_{i}\}$ be the zeros of $g$ such that $r_{1}<|z_{i}|\leq r_{2}$.
		Then
		\begin{equation*}
			\begin{split}
				\frac{1}{2\pi}\int_{0}^{2\pi}\log|g(r_{2}\mathrm{e}^{\mathrm{i}x})|\mathrm{d}x&-\frac{1}{2\pi}\int_{0}^{2\pi}\log|g(r_{1}\mathrm{e}^{\mathrm{i}x})|\mathrm{d}x\\
				&= \sum_{r_{1}<|z_{i}|<r_{2}}\log\frac{r_{2}}{|z_{i}|}+\log\frac{r_{2}}{r_{1}}\frac{1}{2\pi\mathrm{i}}\oint_{|z|=r_{1}}\frac{g'(z)}{g(z)} \mathrm{d}z.
			\end{split}
		\end{equation*}
	\end{lemma}
	Since $f$ is holomorphic in a neighborhood of $\overline{\mathcal{A}(r_{1},r_{2})}$, then there exists $\delta>0$ such that in $\mathcal{A}(r_{1}-\delta,r_{2}+\delta)$,
	\begin{equation*}
		f(z)=h(z)\prod_{r_{1}< |z_{i}|\leq r_{2}}(z-z_{i}) \prod_{|z_{i}|=r_{1}}(z-z_{i})=:g(z)\prod_{|z_{i}|=r_{1}}(z-z_{i})
	\end{equation*}
	for some holomorphic function $h$ defined in $\mathcal{A}(r_{1}-\delta,r_{2}+\delta)$ without zeros and holomorphic function $g$ that has no zeros on the circle $|z|=r_{1}$. 
	
	Applying Lemma \ref{spencer} to $g(z)$ shows that
	\begin{equation}\label{apply}
		\begin{split}
			\frac{1}{2\pi}\int_{0}^{2\pi}\log&|f(r_{2}\mathrm{e}^{\mathrm{i}x})|\mathrm{d}x-\frac{1}{2\pi}\int_{0}^{2\pi}\log|f(r_{1}\mathrm{e}^{\mathrm{i}x})|\mathrm{d}x\\
			&= \sum_{r_{1}<|z_{i}|<r_{2}}\log\frac{r_{2}}{|z_{i}|}+\log\frac{r_{2}}{r_{1}}\frac{1}{2\pi\mathrm{i}}\oint_{|z|=r_{1}}\frac{g'(z)}{g(z)} \mathrm{d}z+\sum_{|z_{i}|=r_{1}}\log\frac{r_{2}}{r_{1}},
		\end{split}
	\end{equation}
	where we use the well-known formula 
	\begin{equation*}
		\frac{1}{2\pi}\int_{0}^{2\pi}\log|a-b\mathrm{e}^{\mathrm{i}x}|\mathrm{d}x=\max\{\log|a|,\log|b|\}.
	\end{equation*}
	Recall the principle value integral that
	\begin{equation}\label{pvi}
		\frac{1}{2\pi\mathrm{i}}\oint_{|z|=r_{1}}\frac{f'(z)}{f(z)} \mathrm{d}z-\frac{1}{2\pi\mathrm{i}}\oint_{|z|=r_{1}}\frac{g'(z)}{g(z)} \mathrm{d}z=\sum_{|z_{i}|=r_{1}}\frac{1}{2\pi\mathrm{i}}\oint_{|z|=r_{1}} \frac{1}{z-z_{i}} \mathrm{d}z=\frac{1}{2}\#\{|z_{i}|=r_{1}\}.
	\end{equation}
	Combining \eqref{pvi} with \eqref{apply} gives \eqref{jens1}.
	
On other hand, 	by noting the principle value integral that
	\begin{equation*}
		\frac{1}{2\pi\mathrm{i}}\oint_{|z|=r_{2}}\frac{f'(z)}{f(z)} \mathrm{d}z-\frac{1}{2\pi\mathrm{i}}\oint_{|z|=r_{1}}\frac{f'(z)}{f(z)} \mathrm{d}z=\frac{1}{2}\#\{ z_{i}\in\partial \mathcal{A}(r_1,r_2)\}+ \#\{ z_{i}\in \mathcal{A}(r_1,r_2)\},
	\end{equation*}
then \eqref{jens2} follows immediately from \eqref{jens1}.
\end{proof}

\subsection{Quantization of winding number}

\begin{proof}[Proof of Theorem \ref{windacc}]
	We only need to prove 	$\nu^{+}(E,\mathrm{i}y)= \omega^{+}(E,\mathrm{i}y),$ the proof of the other case is similar. 
Fix $|y|<h$ such that $L(E,\mathrm{i}y)>0$. By Theorem \ref{regular} and Theorem \ref{cuh}, there exists $h'\in (0,h)$ such that $L(E,\mathrm{i}(y+\epsilon))>0$ for all $0<\epsilon<h'$ and 	
\begin{equation}\label{affine}
	\omega^{+}(E,\mathrm{i}y)=\frac{L(E,\mathrm{i}y_{2})-L(E,\mathrm{i}y_{1})}{y_{2}-y_{1}}, \  \text{if} \ [y_{1},y_{2}]\subset (y,y+h').	
\end{equation}
Note direct computation shows that 
\begin{equation*}
	\nu_{n}(E,\mathrm{i}\tilde{y})=\frac{-1}{2\pi\mathrm{i}n}\oint_{|z|=\mathrm{e}^{-\tilde{y}}}\frac{f_{n}'(z)}{f_{n}(z)} \mathrm{d}z,\quad \text{for} \ |\tilde{y}|<h.
\end{equation*}
By  \eqref{jens2} of Proposition \ref{jensen}, if $[y_{1},y_{2}]\subset (y,y+h')$, then
	\begin{equation*}
		\nu_{n}(E,\mathrm{i}y_{1})\leq  \frac{-1}{2\pi n(y_{2}-y_{1})}\int_{0}^{2\pi}\bigg(\log|f_{n}(\mathrm{e}^{\mathrm{i}(x+\mathrm{i}y_{1})})|-\log|f_{n}(\mathrm{e}^{\mathrm{i}(x+\mathrm{i}y_{2})})|\bigg)\mathrm{d}x\leq \nu_{n}(E,\mathrm{i}y_{2}).
	\end{equation*}
	Applying Theorem \ref{converge} to the above inequality gives
	\begin{equation}\label{supnu}
		\limsup_{n\rightarrow\infty}\nu_{n}(E,\mathrm{i}y_{1})\leq \frac{L(E,\mathrm{i}y_{2})-L(E,\mathrm{i}y_{1})}{y_{2}-y_{1}},
	\end{equation}
	and 
	\begin{equation}\label{infnu}
		\liminf_{n\rightarrow\infty}\nu_{n}(E,\mathrm{i}y_{2})\geq \frac{L(E,\mathrm{i}y_{2})-L(E,\mathrm{i}y_{1})}{y_{2}-y_{1}}.
	\end{equation}

	Let $\epsilon\in(0,h'/10)$. Choosing $y_{1}=y+2\epsilon$ and $y_{2}=y+3\epsilon$ in \eqref{supnu} and using \eqref{affine} gives
	\begin{equation}\label{limsup}
		\limsup_{n\rightarrow\infty}\nu_{n}(E,\mathrm{i}(y+2\epsilon))\leq \frac{1}{\epsilon}(L(E,\mathrm{i}(y+3\epsilon))-L(E,\mathrm{i}(y+2\epsilon)))=\omega^{+}(E,\mathrm{i}y).
	\end{equation}
On the other hand, choosing $y_{1}=y+\epsilon$ and $y_{2}=y+2\epsilon$ in \eqref{infnu} and using \eqref{affine} gives
	\begin{equation}\label{liminf}
		\liminf_{n\rightarrow\infty}\nu_{n}(E,\mathrm{i}(y+2\epsilon))\geq \frac{1}{\epsilon}(L(E,\mathrm{i}(y+2\epsilon))-L(E,\mathrm{i}(y+\epsilon)))=\omega^{+}(E,\mathrm{i}y).
	\end{equation}
	Combine \eqref{limsup} with \eqref{liminf}, we have
	\begin{equation*}
		\nu^{+}(E,\mathrm{i}y)=\lim_{\epsilon\rightarrow0^{+}}\lim_{n\rightarrow\infty}\nu_{n}(E,\mathrm{i}(y+2\epsilon))=\omega^{+}(E,\mathrm{i}y).
	\end{equation*}

	If $E\notin\Sigma_{v,y}$, then $(\alpha,S_{E,v}(\cdot+\mathrm{i}y))\in \mathcal{UH}$ according to Proposition \ref{equi}. By Theorem \ref{regular}, we have $\omega^{+}(E,\mathrm{i}y)= \omega^{-}(E,\mathrm{i}y)$ and thus the proof finishes. 
\end{proof}
\subsection{Application of winding number}
In this subsection, we give the applications of winding number:
\begin{proof}[Proof of Corollary \ref{amo}]
	Avila \cite{Av0}
	proved that for any $E\in\mathbb{C}$, any $y\geq 0$, 
	\begin{equation}\label{amo-ly}
		L(E, \mathrm{i}y) = \max\{\log \lambda + y, L(E,\mathrm{i}0)\}.
	\end{equation}
Then for any  $E\notin \Sigma_{2 \lambda \cos,0},$ the Schr\"odinger cocycle $(\alpha, S_{E,2\lambda\cos}(x))\in  \mathcal{UH}$, in particular $L(E,\mathrm{i}0)>0$,  and it is easy to see the turning point being $\gamma_{*}=L(E,\mathrm{i}0)- \log \lambda$. Hence  the result follows directly from \eqref{amo-ly} and Theorem \ref{windacc}.
\end{proof}

\begin{proof}[Proof of Corollary \ref{zeros}]		
Denote the number of zeros (counted with multiplicities) by
	\begin{equation*}
		\widehat{N}_{n}(r_{1},r_{2}):=\#\{z\in\mathcal{Z}_{f_{n}}: z\in\mathcal{A}(r_{1},r_{2})\}+\frac{1}{2}\#\{z\in\mathcal{Z}_{f_{n}}: z\in\partial \mathcal{A}(r_{1},r_{2})\}.
	\end{equation*}
	By  \eqref{count}, it follows that
	$$\nu_{n}(E,\mathrm{i}y_{2})-\nu_{n}(E,\mathrm{i}y_{1})=	\frac{1}{n} \widehat{N}_{n}( e^{-y_2} ,e^{-y_1}) .$$
	On the other hand, it is obvious that for $y_{1}\leq y_{2}$,
	\begin{equation*}
		\widehat{N}_{n}(\mathrm{e}^{-(y_{2}+\epsilon)},\mathrm{e}^{-(y_{1}-\epsilon)})\leq N_{n}(\mathrm{e}^{-(y_{2}+2\epsilon)},\mathrm{e}^{-(y_{1}-2\epsilon)}) \leq \widehat{N}_{n}(\mathrm{e}^{-(y_{2}+3\epsilon)},\mathrm{e}^{-(y_{1}-3\epsilon)}).
	\end{equation*}
	Hence Corollary \ref{zeros} follows directly from Theorem \ref{windacc}.
\end{proof}

\section{Density of states}
In this section, we investigate the density of states measure of the non-self-adjoint quasi-periodic Schr\"odinger operators and establish Thouless formula. 
Recall the following famous Widom's lemma.
\begin{lemma}[Widom's lemma \cite{Wi90,Wi94}]\label{widom}
	Suppose that  $\{\mathrm{d}\mu_{n}\}_{n\geq 1}$ has compact support in $\mathbb{C}$.
	If $\displaystyle p_{n}(z)=\int_{\mathbb{C}} \log |\zeta-z| \mathrm{d}\mu_{n}(\zeta)$ converges to $p(z)$ almost everywhere in $\mathbb{C}$, then $p$ is local integrable, $\Delta p\geq 0$, and $\mathrm{d}\mu_{n}$ converges weakly to $\frac{1}{2\pi} (\Delta p ) \mathrm{d}m$ as $n\rightarrow \infty$.
\end{lemma}

Now we can finish the proof of Theorem \ref{open}.
\begin{proof}[Proof of Theorem \ref{open}]
	Fix $y\in\mathbb{R}^{d}$ with $|y_{j}|<h(1\leq j\leq d)$. Define the set 
	\begin{equation*}
		\mathcal{K}=\{(E,x): \lim_{n\rightarrow \infty}\frac{1}{n} \log |f_{n}(E,x+\mathrm{i}y)| =L(E,\mathrm{i}y)\},
	\end{equation*}
	and let $\mathcal{L}=(\mathbb{C}\times \mathbb{T}^{d}) \backslash \mathcal{K}$ be the bad set and the corresponding sections are $\mathcal{L}_{x}=\{E:(E,x)\in \mathcal{L}\}$, $\mathcal{L}_{E}=\{x:(E,x)\in \mathcal{L}\}$.	For any fixed $E\in \mathbb{C}\backslash \mathcal{Z}_{L,y}$, we have ${\rm mes}(\mathcal{L}_{E})=0$ by Theorem \ref{converge}. Then it follows from Fubini theorem that
	\begin{equation*}
		|\mathcal{L}| = \frac{1}{(2\pi)^d}\int_{\mathbb{T}^{d}}m(\mathcal{L}_{x}) \mathrm{d}x = \int_{\mathbb{C}} \mathrm{mes}(\mathcal{L}_{E})\mathrm{d}m(E)=0,
	\end{equation*}
	where $|\cdot|$ stands for the product measure on $\mathbb{C}\times \mathbb{T}^{d}$. Therefore for a.e. $x\in\mathbb{T}^d$, 
	\begin{equation*}
		\lim_{n\rightarrow \infty}\frac{1}{n} \log |f_{n}(E,x+\mathrm{i}y)| =L(E,\mathrm{i}y)
	\end{equation*}
	for a.e. $E\in\mathbb{C}$.
	Since $H_{n}(x+\mathrm{i}y)$ is bounded for any $x\in \mathbb{T}^{d}$,  $\mathrm{d} \mathcal{N}_{n}^{x+\mathrm{i}y}$ has compact support in $\mathbb{C}$. 
	Thus Theorem \ref{open}\ref{open1} follows from Lemma \ref{widom} with 
	\begin{equation*}
		p_{n}(E)=\frac{1}{n} \log |f_{n}(E,x+\mathrm{i}y)|,\quad \mathrm{d}\mu_{n}=\mathrm{d}\mathcal{N}_{n}^{x+\mathrm{i}y},\quad p(E)=L(E,\mathrm{i}y).
	\end{equation*}
	
To prove Theorem \ref{open}\ref{open2}, we recall the following:	
\begin{lemma}[Lower envelope theorem \cite{Lan}]\label{schwartz}
	Suppose that $\{\mathrm{d}\mu_{n}\}_{n\geq 1}$ has compact support in $\mathbb{C}$. Denote $\displaystyle p_{n}(z)=\int_{\mathbb{C}} \log |z-\zeta|\mathrm{d}\mu_{n}(\zeta)$.
	Assume that $\mathrm{d}\mu_{n} \rightharpoonup \mathrm{d}\mu$, as $n\rightarrow\infty$. Then
	\begin{equation*}
		\limsup_{n\rightarrow\infty} p_{n}(z)= \int_{\mathbb{C}} \log |z-\zeta|\mathrm{d}\mu(\zeta)\quad\text{ for a.e. } z\in\mathbb{C}.
	\end{equation*}
\end{lemma}
Theorem \ref{open}\ref{open2} follows from Lemma \ref{schwartz} and  the weak uniqueness of subharmonic function.
\end{proof}

As a consequence of the non-self-adjoint Thouless formula, we can establish the relation between the winding number and the density of states measure for the one-frequency quasi-periodic Schr\"odinger operators. Let us first prove a lemma that $m(\mathcal{Z}_{L,y})=0$ for any Schr\"odinger operators with real analytic potential $v(\cdot+\mathrm{i}y)$.
\begin{lemma}\label{zeroset}
	Let $(\alpha,v)\in\mathbb{R}\backslash \mathbb{Q}\times C^{\omega}(\mathbb{T}_{h},\mathbb{R})$. The following hold.
	\begin{enumerate}[font=\normalfont, label={(\arabic*)}]
		\item  \label{item:zeroacc}For any $E\notin\mathbb{R}$, 
		$\omega^{+}(E,\mathrm{i}0)=\omega^{-}(E,\mathrm{i}0)=0.$
	
		\item \label{item:zeroset}For any $|y|<h$,
		$ \mathcal{Z}_{L,y}\subset \mathbb{R}.$
	\end{enumerate}	
\end{lemma}
\begin{proof}
	Let us prove Lemma \ref{zeroset}\ref{item:zeroacc}. For any $E\in\mathbb{C}\backslash\mathbb{R}$, the Schr\"odinger cocycle $(\alpha,S_{E,v})\in \mathcal{UH}$. It follows from Theorem \ref{regular} and Theorem \ref{cuh} that there exists $h'\in(0,h)$ such that for any $|\epsilon|<h'$
	\begin{equation}\label{leftright}
		\omega^{+}(E,\mathrm{i}\epsilon)=\omega^{-}(E,\mathrm{i}\epsilon).
	\end{equation}

Apply Proposition \ref{jensen} with $r_{1}=1$ and $r_{2}=\mathrm{e}^{\epsilon}$ with $\epsilon\in(0,h'/10)$, we have
\begin{equation}\label{jensen1}
	\frac{1}{2\pi\mathrm{i}n}\oint_{|z|=1}\frac{f'_{n}(z)}{f_{n}(z)} \mathrm{d}z\leq	\frac{1}{2\pi n\epsilon}\int_{0}^{2\pi}\bigg(\log|f_{n}(\mathrm{e}^{\mathrm{i}(x-\mathrm{i}\epsilon)})|-\log|f_{n}(\mathrm{e}^{\mathrm{i}x})|\bigg)\mathrm{d}x.
\end{equation}
Note that $H_{n}(x)$ is self-adjoint whose spectrum $\{E_{j}(x)\}$ is real, which means
\begin{equation}\label{vanish}
	\frac{1}{2\pi\mathrm{i}n}\oint_{|z|=1}\frac{f'_{n}(z)}{f_{n}(z)} \mathrm{d}z=\frac{1}{n}\sum_{j=1}^{n}\frac{-1}{2\pi}\int_{0}^{2\pi} \frac{\partial}{\partial x} \arg(E_{j}(x)-E)\mathrm{d}x=0.
\end{equation}
It follows from \eqref{jensen1}, \eqref{vanish}, and Theorem \ref{converge} that
\begin{equation}\label{leqzero}
	\begin{split}
		0&\leq \lim_{n\rightarrow \infty}\frac{1}{2\pi n\epsilon}\int_{0}^{2\pi}\bigg(\log|f_{n}(\mathrm{e}^{\mathrm{i}(x-\mathrm{i}\epsilon)})|-\log|f_{n}(\mathrm{e}^{\mathrm{i}x})|\bigg)\mathrm{d}x\\
		&= \frac{1}{\epsilon}(L(E,-\mathrm{i}\epsilon)-L(E,\mathrm{i}0))=-\omega^{-}(E,\mathrm{i}0).
	\end{split}
\end{equation}

On the other hand, apply Proposition \ref{jensen}  with $r_{2}=1$ and $r_{1}=\mathrm{e}^{\epsilon}$ with $\epsilon\in(-h'/10,0)$, we have
\begin{equation}\label{jensen2}
	\frac{1}{2\pi\mathrm{i}n}\oint_{|z|=1}\frac{f_{n}'(z)}{f_{n}(z)} \mathrm{d}z \geq\frac{1}{2\pi n\epsilon}\int_{0}^{2\pi}\bigg(\log|f_{n}(\mathrm{e}^{\mathrm{i}x})|-\log|f_{n}(\mathrm{e}^{\mathrm{i}(x-\mathrm{i}\epsilon)})|\bigg)\mathrm{d}x.
\end{equation}
Hence combine \eqref{jensen2} with \eqref{vanish} and Theorem \ref{converge} we have
\begin{equation}\label{geqzero}
	\begin{split}
		0&\geq \lim_{n\rightarrow \infty}\frac{1}{2\pi n\epsilon}\int_{0}^{2\pi}\bigg(\log|f_{n}(\mathrm{e}^{\mathrm{i}x})|-\log|f_{n}(\mathrm{e}^{\mathrm{i}(x+\mathrm{i}\epsilon)})|\bigg)\mathrm{d}x\\
		&=\frac{1}{\epsilon}(L(E,\mathrm{i}0)-L(E,\mathrm{i}\epsilon))=-\omega^{+}(E,\mathrm{i}0).
	\end{split}
\end{equation}
Thus Lemma \ref{zeroset}\ref{item:zeroacc} follows from \eqref{leftright}, \eqref{leqzero}, and \eqref{geqzero}.

As for Lemma \ref{zeroset}\ref{item:zeroset}, since $L(E,\mathrm{i}y)$ is convex and $L(E,\mathrm{i}0)> 0$ for $E\notin \mathbb{R}$, we deduce that $L(E,\mathrm{i}y)\geq L(E,\mathrm{i}0)$ for any $|y|<h$ by Lemma \ref{zeroset}\ref{item:zeroacc}, which means $\mathcal{Z}_{L,y}\subset \mathbb{R}$.
\end{proof}

\begin{proof}[Proof of Corollary \ref{relation}]
	According to Avila's global theory \cite{Av0}, $L(E,\mathrm{i}y)$ is convex and piecewise linear to $y$ with right derivatives $\omega^{+}(E,\mathrm{i}{y})$. 
	Thus
	\begin{equation*}
		L(E,\mathrm{i}y)=\int_{0}^{y} \omega^{+}(E,\mathrm{i}\tilde{y}) \mathrm{d}\tilde{y}+L(E,\mathrm{i}0).
	\end{equation*}
Note if  $E\in \mathbb{R}$,  since $L(E,\mathrm{i}y)$ is convex and even with respect to $y$,  it follows that 
\begin{equation}\label{posi} L(E,\mathrm{i}y)\geq L(E,\mathrm{i}0)>0,  \quad \text{for any}  \ y\in (0,h).
\end{equation}
  If $E\notin \mathbb{R}$, by Lemma \ref{zeroset}\ref{item:zeroacc}, we have $\omega^{+}(E,\mathrm{i}0)=0$, then \eqref{posi} again follows from convexity. 
Therefore, we can apply  Theorem \ref{windacc} and obtain
\begin{equation}\label{lhs}
	L(E,\mathrm{i}y)=\int_{0}^{y} \nu^{+}(E,\mathrm{i}\tilde{y}) \mathrm{d}\tilde{y}+L(E,\mathrm{i}0).
\end{equation}

Meanwhile, by Lemma \ref{zeroset}\ref{item:zeroset}, one can apply Theorem \ref{open}\ref{open2} to obtain that
	\begin{equation}\label{rhs}
		L(E,\mathrm{i}y)= \int_{\mathbb{C}} \log |\tilde{E}-E|\mathrm{d}\mathcal{N}^{\mathrm{i}y}(\tilde{E}).
	\end{equation}
So both two sides in \eqref{accden} are equal by \eqref{lhs} and \eqref{rhs}, thus Corollary \ref{relation} is true.
\end{proof}

\appendix
\section{Proof of Lemma \ref{three}}\label{app}
Recall that for any $A\in {\rm SL}(2,\mathbb{C})$ there exist unit vectors $u_{A}^{+}, u_{A}^{-}, v_{A}^{+}, v_{A}^{-}$ such that $Au_{A}^{+} =\|A\| v_{A}^{+}$, $Au_{A}^{-}=\|A\|^{-1} v_{A}^{-}$. Moreover, $|u_{A}^{+}\wedge u_{A}^{-}|= |v_{A}^{+} \wedge v_{A}^{-}|=1$. 

Now we finish the proof of Lemma \ref{three}. 
Let us first show for sufficiently large $K$,
\begin{equation}\label{inequ1}
	{\rm mes}(\mathcal{F}_{K}\cap \mathcal{G}_{K}) \leq {\rm e}^{-(K^{1-\Cr{bousigh}})^{\frac{1}{2d}}}.
\end{equation}
For any $x\in \mathcal{F}_{K} \cap \mathcal{G}_{K}$, by Corollary \ref{ldtexten} and the definition of $\mathcal{F}_{K}$,
\begin{equation}\label{fupper}
	\begin{split}
		&|f_{n}(x+\mathrm{i}y)| \leq {\rm e}^{nL_{n}(E,\mathrm{i}y)-5C_{\star}nK^{-\Cr{bousigh}}},\\
		&\|M_{n}(x+\mathrm{i}y)\| \geq {\rm e}^{nL_{n}(E,\mathrm{i}y)-\Cr{bouCh}nK^{-\Cr{bousigh}}}.
	\end{split}
\end{equation}
Since $ f_{n}(x+\mathrm{i}y)=M_{n}(x+\mathrm{i}y)\vec{e}_{1} \wedge \vec{e}_{2}$ and  $ |\langle u_{n}^{+}(x+\mathrm{i}y),\vec{e}_{1}\rangle|= |u_{n}^{-}(x+\mathrm{i}y)\wedge \vec{e}_{1}|$,
then
\begin{equation}\label{flower}
	\begin{split}
		|f_{n}(x+\mathrm{i}y)| \geq &\|M_{n}(x+\mathrm{i}y)\| |u_{n}^{-}(x+\mathrm{i}y)\wedge \vec{e}_{1}| |v_{n}^{+}(x+\mathrm{i}y)\wedge \vec{e}_{2}|\\ &-\|M_{n}(x+\mathrm{i}y)\|^{-1} |u_{n}^{+}(x+\mathrm{i}y)\wedge \vec{e}_{1}| |v_{n}^{-}(x+\mathrm{i}y)\wedge \vec{e}_{2}|.
	\end{split}
\end{equation}
Combine (\ref{fupper}) with (\ref{flower}) and $C_{\star}\geq \Cr{bouCh}$ one can get that
\begin{equation*}
	|u_{n}^{-}(x+\mathrm{i}y)\wedge \vec{e}_{1}| \cdot|v_{n}^{+}(x+\mathrm{i}y)\wedge \vec{e}_{2}| \leq {\rm e}^{-4C_{\star}n K^{-\Cr{bousigh}}}.
\end{equation*}
So either
\begin{equation}\label{either1}
	|u_{n}^{-}(x+\mathrm{i}y)\wedge \vec{e}_{1}| \leq  {\rm e}^{-2 C_{\star}nK^{-\Cr{bousigh}}}
\end{equation}
or
\begin{equation}\label{either2}
	|v_{n}^{+}(x+\mathrm{i}y)\wedge \vec{e}_{2}|  \leq {\rm e}^{-2 C_{\star}nK^{-\Cr{bousigh}}}.
\end{equation}

Fix any $x\in \mathcal{F}_{K}\cap \mathcal{G}_{K}$. By the preliminary discussion, among the three terms in $\mathcal{F}_{K}$ either two terms satisfy (\ref{either1}) or two terms satisfy (\ref{either2}) according to the pigeonhole principle. The following result was essentially proved in \cite{GS08, BV13}.
\begin{lemma}\label{wedgebug}
	If (\ref{either1}) occurs twice in $\mathcal{F}_{K}$, then
	\begin{equation*}
		\begin{split}
			|\vec{e}_{1}\wedge M_{j_{2}-j_{1}}(T^{j_{1}}x+\mathrm{i}y) \vec{e}_{1}|&\leq |u_{n}^{-}(T^{j_{2}}x+\mathrm{i}y)\wedge M_{j_{2}-j_{1}}(T^{j_{1}}x+\mathrm{i}y) u_{n}^{-}(T^{j_{1}}x+\mathrm{i}y)|\\
			&\quad + \sqrt{2}\|M_{j_{2}-j_{1}}(T^{j_{1}}x+\mathrm{i}y)^{-1}\|\cdot |\vec{e}_{1}\wedge u_{n}^{-}(T^{j_{1}}x+\mathrm{i}y)|\\
			&\quad +\sqrt{2}\|M_{j_{2}-j_{1}}(T^{j_{1}}x+\mathrm{i}y)\|\cdot |\vec{e}_{1}\wedge u_{n}^{-}(T^{j_{2}}x+\mathrm{i}y)|.
		\end{split}
	\end{equation*}
	If (\ref{either2}) occurs twice in $\mathcal{F}_{K}$, then
	\begin{equation*}
		\begin{split}
			|\vec{e}_{2}\wedge M_{j_{2}-j_{1}}(T^{n+j_{1}}x+\mathrm{i}y) \vec{e}_{2}|&\leq |v^{+}_{n}(T^{j_{2}}x+\mathrm{i}y) \wedge M_{j_{2}-j_{1}}(T^{n+j_{1}}x+\mathrm{i}y)v^{+}_{n}(T^{j_{1}}x+\mathrm{i}y)|\\
			&\quad+  \sqrt{2} \|M_{j_{2}-j_{1}}(T^{n+j_{1}}x+\mathrm{i}y)\|\cdot |\vec{e}_{2}\wedge v_{n}^{+}(T^{j_{2}}x+\mathrm{i}y)| \\
			&\quad + \sqrt{2}\|M_{j_{2}-j_{1}}(T^{n+j_{1}}x+\mathrm{i}y)^{-1}\|\cdot |\vec{e}_{2}\wedge v_{n}^{+}(T^{j_{1}}x+\mathrm{i}y)|.
		\end{split}
	\end{equation*}
	Moreover, the same type of estimates hold if  $(j_{1},j_{2})$ is replaced by $(0,j_{1})$ or $(0,j_{2})$.
\end{lemma}
By Lemma \ref{wedgebug}, \eqref{either1}, \eqref{either2}, for sufficient large $K$, we have
\begin{equation}\label{condict}
	\begin{split}
		&|f_{j_{2}-j_{1}-1}(T^{j_{1}}x+\mathrm{i}y)|=|\vec{e}_{1}\wedge M_{j_{2}-j_{1}}(T^{j_{1}}x+\mathrm{i}y) \vec{e}_{1}| <{\rm e}^{-C_{\star}nK^{-\Cr{bousigh}}},\\
		&|f_{j_{2}-j_{1}-1}(T^{n+j_{1}+1} x+\mathrm{i}y)|=|\vec{e}_{2}\wedge M_{j_{2}-j_{1}}(T^{n+j_{1}} x+\mathrm{i}y) \vec{e}_{2}|<{\rm e}^{-C_{\star}nK^{-\Cr{bousigh}}}.
	\end{split}
\end{equation}

Suppose \eqref{inequ1} does not hold, then ${\rm mes}(\mathcal{F}_{K}\cap \mathcal{G}_{K}) >{\rm e}^{-(K^{1-\Cr{bousigh}})^{\frac{1}{2d}}}$. 
By \eqref{condict}, we deduce that for the choice of $l$ from $j_{1}-1, j_{2}-1, j_{2}-j_{1}-1$,
\begin{equation*}
	{\rm mes}\{x\in \mathbb{T}^{d}:|f_{l}(x+\mathrm{i}y
	)|<  {\rm e}^{-C_{\star} nK^{-\Cr{bousigh}}}\} \geq  {\rm mes}(\mathcal{F}_{K}\cap \mathcal{G}_{K})>{\rm e}^{-(K^{1-\Cr{bousigh}})^{\frac{1}{2d}}}.
\end{equation*}
Recall that $n>K\delta^{-1}$, we note that for any $l<K^{\frac{1-\Cr{bousigh}}{2}}$,
\begin{equation*}
	{\rm e}^{-(K^{1-\Cr{bousigh}})^{\frac{1}{2d}}} > \mathrm{e}^{-(nK^{-\Cr{bousigh}})^{\frac{1}{2d}}}>{\rm e}^{- (nK^{-\Cr{bousigh}} l^{-1})^{\frac{1}{d}}}.
\end{equation*}
Hence we have $
	{\rm mes}\{x\in \mathbb{T}^{d}:|f_{l}(x+\mathrm{i}y)|< {\rm e}^{-C_{\star} nK^{-\Cr{bousigh}}}\}>{\rm e}^{- (nK^{-\Cr{bousigh}} l^{-1})^{\frac{1}{d}}}. $
Apply Lemma \ref{GS08} to $u(x)=\log|f_{l}(x+\mathrm{i}y)|$ with
\begin{equation*}
	\mathcal{M}=C_{E}^{v}l,\quad \mathcal{S}=C_{\star}nK^{-\Cr{bousigh}},\quad\varrho= {\rm e}^{-(nK^{-\Cr{bousigh}} l^{-1})^{\frac{1}{d}}}.
\end{equation*}
One can get that for $l\in (l_{0}, K^{\frac{1-\sigma_{3}}{2}})$ with $K$ sufficiently large,
\begin{equation*}
	\sup_{x\in \mathbb{T}^{d}} u(x)\leq C_{h-h',d} \mathcal{M}-\frac{\mathcal{S}}{C_{h-h',d} \log^{d} (C_{d}/\varrho)} \leq -\sqrt{C_{\star}} l.
\end{equation*}
Then it follows from Lemma \ref{closeone} that $\sup_{x\in\mathbb{T}^{d}}\|M_{l}(x+\mathrm{i}y)\|<6+2(|E|+\|v\|_{h})$, which contradicts to the assumption $L(E,\mathrm{i}y)>\gamma$. 

We are done with the proof by choosing $\Cr{trisig}=\min\{\frac{\Cr{bousigh}}{2}, \frac{1-\Cr{bousigh}}{4d}\}$ since for  large enough $K$,
\begin{equation*}
		{\rm mes} (\mathcal{F}_{K})
		\leq {\rm mes}(\mathbb{T}^{d}\backslash \mathcal{G}_{K})+  {\rm mes} (\mathcal{F}_{K}\cap \mathcal{G}_{K}) 
		\leq K{\rm e}^{-\Cr{bouCh}K^{\Cr{bousigh}}} + {\rm e}^{-(K^{1-\Cr{bousigh}})^{\frac{1}{2d}}}\leq \mathrm{e}^{-K^{\Cr{trisig}}}.
\end{equation*}
\qed

\section*{Acknowledgements}
X. Wang would like to thank M. Goldstein for useful discussions on Cantan estimates. Q. Zhou also would like to thank S. Jitomirskaya for  inspired discussions. 
This work was partially supported by National Key R\&D Program of China (2020 YFA0713300) and Nankai Zhide Foundation.  J. You was also partially supported by NSFC grant (11871286). Q. Zhou was supported by NSFC grant (12071232), the Science Fund for Distinguished Young Scholars of Tianjin (No. 19JCJQJC61300).

\end{document}